\documentclass[12pt]{amsart}
\usepackage[utf8]{inputenc}
\usepackage{amssymb,amsmath,amsfonts,latexsym,verbatim,mathtools,tabularx,booktabs,caption}
\usepackage[dvipsnames]{xcolor}

\usepackage{tikz}
\usetikzlibrary{decorations.pathreplacing, arrows.meta}

\usepackage{pgfplots}
\pgfplotsset{compat=1.18}

\usepackage{subcaption}

\usepackage{hyperref}

\hypersetup{
pdfborderstyle={},
allcolors=accentcolor,
    colorlinks=true,
    linkcolor=blue,
    citecolor=blue,
    filecolor=magenta,      
    urlcolor=cyan,
    pdftitle={Common ratio effect},
    pdfpagemode=FullScreen,
    }

\def\mnoss{{MNOSS} }
\def\mnossc{ MNOSS}

\def\Re{\mathbf{R}}

\def\E{\mathbf{E}}
\def\Cov{\mathbf{Cov}}
\def\Pr{\mathbf{Pr}}

\def\ep{\varepsilon}

\def\phi{\varphi}

\def\then{\Longrightarrow}

\newcommand{\df}[1]{\textit{\textbf{#1}}}

\newtheorem{theorem}{Theorem}

\newtheorem{proposition}[theorem]{Proposition}
\newtheorem{prps}[theorem]{Proposition}

\theoremstyle{definition}

\theoremstyle{remark}

\newtheorem{remark}{Remark}

\theoremstyle{definition}

\newtheorem{exm}{Example}
\newtheorem{example}[exm]{Example}

\newtheoremstyle{named}{}{}{\itshape}{}{\bfseries}{:}{.5em}{#1\thmnote{#3}}
\theoremstyle{named}

\usepackage{tikz}
\usetikzlibrary{trees}
\usetikzlibrary{patterns}

\usepackage[authoryear]{natbib}\setcitestyle{square}

	\usepackage[left=1.25in,top=1.25in,right=1.25in,bottom=1.25in]{geometry}
	\usepackage[onehalfspacing]{setspace}

\title[Allais by paired choices vs.\ valuations]{Robust Testing Of the Allais Paradox By Paired Choices vs.\ Paired Valuations}
\author[Echenique and Tserenjigmid]{Federico Echenique and Gerelt Tserenjigmid}
\thanks{Echenique: Department of Economics, UC Berkeley, \texttt{fede@econ.berkeley.edu} \\ Tserenjigmid: Department of Economics, UC Santa Cruz,  \texttt{gtserenj@ucsc.edu}.}
\thanks{We are very grateful to Christina McGranaghan, Kirby Nielsen, Ted O'Donoghue, Jason Somerville, and Charlie Sprenger for feedback and advice on an earlier version of our paper.}

\begin{document}

\begin{abstract}
\cite*{mcgranaghan2024distinguishing} show that standard paired choice tests for the common ratio effect are structurally biased when choice is stochastic, proposing valuation tests as a robust alternative. Using valuation tests, they find no systematic evidence for the common ratio effect, seemingly overturning much of the extant literature. We evaluate this conclusion in light of stochastic choice theory. We argue that valuation tests are inherently biased and lack predictive power under standard expected utility assumptions. In contrast, we advocate for a ``strong'' paired choice test, proving it remains robustly unbiased across common models of stochastic choice. Applying this strong test to existing experimental data, we find that the common ratio effect remains highly prevalent.
\end{abstract}

\maketitle
\thispagestyle{empty}

\clearpage

\tableofcontents
\thispagestyle{empty}
\setcounter{page}{0}

\newpage

\section{Introduction}

Few behavioral phenomena have had the impact of the \cite{allais1953comportement} paradox. It has sparked a large empirical literature and motivated many theoretical models of non-expected utility \citep{machina1987choice,camerer1995individual,machina2008non, barberis2013thirty, machina2018non}. We are concerned with a basic aspect of the Allais paradox, the common ratio effect. In particular, we focus on the meaning of the common ratio effect when there is noise or randomness in agents' choices. Such randomness seems inevitable in any serious empirical discussion of the Allais paradox. 

To be concrete, consider two choice problems. First, a choice between a sure payment $A$ and a risky lottery $B$. Second, a choice between $C$ and $D$, which are scaled-down versions of $A$ and $B$ involving the same probabilities. An expected utility maximizer chooses $A$ over $B$ if and only if they choose $C$ over $D$. The Allais paradox is the systematic violation of this prediction. The \df{common ratio effect} is the behavior of choosing $A$ over $B$ and $D$ over $C$. A test for whether the probability of choosing $A$ over $B$ equals the probability of choosing $C$ over $D$ is called the \df{weak paired choice test}.

\cite{Loomes_2005} and \cite*{mcgranaghan2024distinguishing} show that randomness in choices may induce the common ratio effect even for an expected utility agent. In other words, the weak paired choice test that defines the common ratio effect is biased against expected utility when choices are random.

This bias has motivated \cite{mcgranaghan2024distinguishing} (MNOSS) to propose \textbf{paired valuation tests} instead of paired choices (essentially eliciting certainty equivalents, or scaled certainty equivalents, for lotteries $B$ and $D$). There is a mean and sign version of the valuation test (see Section~\ref{sec:valtests}). \mnoss demonstrate that these tests can perform better than the weak paired choice test when choices are stochastic, and run new experiments to elicit valuation data. Their analysis using valuation tests shows little support for the common ratio effect using aggregate data, and ``substantial heterogeneity'' at the individual level. They find that individual evidence for the common ratio effect in valuations is predictive of common ratio effects in paired choices, which validates their methodology. Consequently, they conclude, there is no support for a systematic common ratio effect.

The results in \mnoss thus seem to overturn the majority of the existing empirical evidence, which has overwhelmingly favored the common ratio effect. For example, a recent survey by \cite{Blavatskyy_Panchenko_Ortmann_2023} documents that 78\% of 143 studies testing expected utility theory find evidence consistent with the common ratio effect. 

Our paper contributes to the debate by taking a close look at models of expected utility and stochastic choice. The bias of the weak paired choice test relies on an assumed model of stochastic choice. The test is biased under the stochastic choice model that adds an i.i.d. noise term to expected utility: we shall call it the i.i.d. Additive Random Expected Utility (iAREU) model. While common in the experimental literature, as a random version of expected utility, the iAREU is questionable. Under the iAREU, each realized utility function will, with probability one, be non-expected utility. The independence axiom is violated with probability one. 

A natural alternative model is the random expected utility model of \cite{gul2006random}.\footnote{For an exposition of stochastic choice, including a discussion and comparison of iAREU and random expected utility, see \cite{strzalecki2025stochastic}. We follow his terminology here and repeat some of the points he makes in comparing various models.} The random expected utility model delivers an expected utility function with probability one. Our first observation is that \emph{the weak paired choice test is unbiased} under the random expected utility model. In fact, the main axiom used by Gul and Pesendorfer rules out the common ratio effect in a weak paired choice test. 

If one abandons the weak test, a robust alternative exists. We advocate for a \textbf{strong paired choice test}: If choices are noisy, it is natural to say that an agent prefers $A$ to $B$ if they choose $A$ more than half the time. The strong test simply requires that an EU agent prefers $A$ to $B$ if and only if they prefer $C$ to $D$. We show that the strong test remains unbiased for the different models of random utility that we consider in the paper, including the iAREU.

Next, we study the proposed valuation tests and show that they can be problematic. In fact, we obtain an ``anything goes'' result for the mean valuation test under assumptions that include the iAREU model as a special case. This means that a very wide range of mean valuations is consistent with expected utility theory. For the sign test, we also show that ``anything goes'' once the error terms affecting individual choice can be arbitrarily correlated, even when each error is symmetric and has a mean of zero. See Proposition~\ref{prop:anythingoes} for a formal statement of these results.

\begin{table}[ht]
    \footnotesize
    \centering
    % This redefines the 'X' column to use 'm' (middle) instead of 'p' (top)
    \renewcommand{\tabularxcolumn}[1]{m{#1}}
    
    \begin{tabularx}{\textwidth}{l >{\centering\arraybackslash}X >{\centering\arraybackslash}X >{\centering\arraybackslash}X >{\centering\arraybackslash}X}
        \toprule
         & Weak paired choice test & Strong paired choice test & Mean valuation test & Sign valuation test \\ 
        \midrule
        Random EU & Unbiased & Unbiased & Unbiased & Unbiased \\ \addlinespace
        Fechnerian random utility & Biased & Unbiased & Biased & Unbiased$^*$ \\ \addlinespace
        Random Perception and Utility & & Unbiased$^*$ & Biased & Biased$^*$ \\ \addlinespace
        Random prospect theory & & Unbiased & Biased & Biased \\ \addlinespace
        Assumption 2b & Biased & Unbiased & Anything goes & Anything goes \\ \addlinespace
        Assumption 3 &  & Unbiased & Anything goes & Unbiased \\ \addlinespace
        Weak expected utility & & Unbiased & & \\ 
        \bottomrule
    \end{tabularx}
    \caption{Summary of results}\label{table:summary}
\end{table}

Table~\ref{table:summary} contains a summary of our results. The models listed in Table~\ref{table:summary} are described in detail in Section~\ref{sec:stochchoicemodels} below. The ``Random EU'' model is Gul and Pesendorfer's random expected utility model, while Fechnerian models are well-known generalizations of iAREU. Random perception and random prospect theory assume randomness in probabilities, either as perception or measurement error, or as a random probability weighting function. The row for ``Assumption 2b'' corresponds to the assumptions in \mnossc, under which the weak test is biased. \mnoss show that the sign test is unbiased under different assumptions, which we have called ``Assumption 3'' in the paper. Our strong test is also unbiased under either of these assumptions.

The result on the unbiasedness of the strong paired choice for the model of random prospect theory, which seems to be new, also requires assumptions that are laid out and discussed in our paper. The entries in the table that are marked with an asterisk require assumptions that are described in the paper.

Table~\ref{table:summary} highlights, on the one hand, that all the tests we have considered are unbiased for random expected utility. On the other hand, the strong paired choice test is broadly robust across the different ways in which one may model stochastic choice in the expected utility framework. If we adopt the strong paired test as our test for the common ratio effect, there is substantial support for the Allais paradox in the existing literature. The evidence on the strong paired choice test is discussed in Section~\ref{sec:empirical}. 

Finally, the strongest critique of the weak paired choice tests argues that a wide range of choice frequencies is consistent with expected utility and unrestricted preference shocks. We reproduce this critique in Section~\ref{sec:empirical}. The flip side of this critique is that allowing for such arbitrary noise leads to a test with very low power. Using simulated choices from prospect theory, in Section~\ref{sec:empirical}, we illustrate the problem of low power and compare it with the strong paired test.

\paragraph{\emph{Empirical implications}.} When we look at the data using the strong paired choice test, we find strong empirical support for the common ratio effect. In the  143 experimental studies from the recent meta-analysis by \cite{Blavatskyy_Panchenko_Ortmann_2023},  41\% of these studies display a common ratio effect, while 7\% display a reverse common ratio effect. When weighting these studies by the number of participants, over 50\% of all surveyed experiments exhibit either a common ratio effect or a reverse common ratio effect. 

Applying the strong test to the experimental data in \mnoss reveals a 10\% prevalence of the common ratio effect and a 10\% prevalence of the reverse effect. This lower detection rate reinforces one of their critiques: parameter selection affects results. Historically, the literature has tested for the common ratio effect within a selective region of the parameter space. Our analysis of their data and additional numerical exercises in Section~\ref{sec:empirical} confirms that stepping outside the traditional window diminishes the effect. Arbitrary parameters often fail to induce the behavioral tension of the common ratio effect, even under prospect theory. Understanding how and why specific parameter choices drive these preference reversals, and what the proper strategy should be for testing the common ratio effect, seems like an important question that deserves more attention, which we begin to discuss in Section~\ref{sec:empirical}.

Section~\ref{sec:takeaway} summarizes the main takeaway messages of our paper and discusses potential applications of our strong paired choice test and methodology for detecting different behavioral puzzles, such as \emph{present bias} in intertemporal choice, from choice data.

\section{Preliminaries}\label{sec:preliminaries}

We study choices between lotteries with objective risk. Given is a finite set $X\subset\Re_{+}$ of monetary prizes with $0\in X$. A \df{lottery} is a probability distribution over $X$. We denote by $\Delta(X)$ the set of all lotteries. We focus on simple lotteries with two outcomes: a monetary payoff $x> 0$ and $0$. The lottery that pays $x$ with probability $p$ and $0$ with probability $1-p$ is denoted by $(x,p)$. 

Before stating our results, we briefly review the necessary background notions from stochastic choice theory. Our terminology follows the recent monograph by \cite{strzalecki2025stochastic}. 

\subsection{Stochastic choice.}\label{sec:stochchoicemodels}

A stochastic choice function $\rho$ specifies, for any finite set $D\subseteq \Delta(X)$, the probability of choosing each lottery in $D$. Thus $\rho(\ell \mid D)$ is the probability of choosing $\ell\in D$. We focus on sets $D$ with two lotteries, which allow us to simplify the notation. When $D=\{\ell_A,\ell_B\}$, we write $\rho(\ell_A,\ell_B)$, or just $\rho(A,B)$, for the probability of choosing lottery $\ell_A$ from the doubleton set $\{\ell_A,\ell_B\}$. 

Empirical analysis of choice data often assumes that choice is random. There are at least two sources of random choice: preference heterogeneity and random preferences. First, if the choice data are collected from a population of agents, choices are stochastic due to preference heterogeneity --- different people make different choices when they face the same choice problem. Second, choice data may be stochastic due to random preferences: the same agent makes different choices when they face the same choice problem. The same agent may choose different options due to mistakes, a preference for randomization, and so on.\footnote{Individuals often make different choices in apparently identical situations, even when the interval between successive choices is very short (see \cite{tversky1969intransitivity}, \cite{ballinger1997}, \cite{hey2001does}, and \cite{agranov2017stochastic}.)} 

In deterministic settings, a typical test of expected utility theory is to check the Independence Axiom via paired choice tasks. For example, if $A=(x, 1)$ is chosen over $B=(y, p)$, but $D=(y, r p)$ is chosen over $C=(x, r)$, then we would conclude that expected utility theory is violated in favor of the common ratio effect. The question at hand is how to test expected utility theory (or the common ratio effect) when stochastic choice data is observed, where stochasticity is due to preference heterogeneity and random preferences. In particular, we observe $\rho(A, B)$ and $\rho(C, D)$. Without restrictions on (or a model of) stochastic choice data, such testing is not possible. Below, we review some of the most popular models of stochastic choice. 

\subsubsection{Random utility}
We say that a stochastic choice function $\rho$ is a \df{random utility} choice function if there exists a utility function $v:\Delta(X)\to\Re$, together with a random function $\epsilon:\Delta(X)\to\Re$, such that 
\[
\rho(\ell,\ell')=\Pr \big(v(\ell)+\epsilon(\ell)> v(\ell')+\epsilon(\ell')\big),
\] for all $\ell,\ell'\in\Delta(X)$, where $\Pr$ denotes the probability inherent in the probability model of the random function $\epsilon$.\footnote{When the set of alternatives is finite, which is always true in the experimental setting, the additive structure of error is without loss of generality (see \cite{strzalecki2025stochastic}).}

\subsubsection{Fechnerian random utility}
The stochastic choice function $\rho$ is \df{Fechnerian} if there exists a utility function $v:\Delta(X)\to\Re$, together with a strictly monotonically increasing function $F:\Re\to\Re$ that is symmetric (meaning that $F(x)=1-F(-x)$) such that 
\[
\rho(\ell,\ell')=F\big(v(\ell )-v(\ell')\big).
\] 

A random utility model given by a utility $v$ and a random error $\epsilon$, such that $\epsilon(\ell)$ are i.i.d. draws from some strictly increasing distribution function, is Fechnerian, where $F$ is the cumulative distribution function of the difference $\epsilon(\ell)-\epsilon(\ell')$.\footnote{Since $\rho(\ell_A, \ell_B)+\rho(\ell_B, \ell_A)=1$, we need to have $F(x)=1-F(-x)$. This implies $F(0)=\frac{1}{2}$.} In this case, we shall always assume (in fact, without loss of generality) that $\E[\epsilon(\ell)]=0$.

\subsubsection{i.i.d. Additive Random Expected Utility (iAREU)}\label{sec:iAREU}

A special Fechnerian model obtains when the utility function $v$ takes the expected utility form. Let the von Neumann-Morgenstern (vNM) utility function $u:X\to\Re$ be such that $v(\ell)=\E_{\ell} u(x)$; the latter expectation is taken over the random prizes in the lottery $\ell$. We shall always normalize $u$ so that $u(0)=0$ and assume that $u$ is continuous and strictly monotonically increasing. In a notational shortcut, we write $\E_{\ell} u(x)$ as $\E[u(\ell)]$.

If $\ell$ is the simple lottery $(x,p)$, then $v(\ell)=\E[u(\ell)]=p\,u(x)$. Once $v$ has an expected utility representation, the assumption that $\E[\epsilon]=0$ is restrictive. 

Note that, if the errors $\epsilon(\cdot)$ are independent and identically distributed (i.i.d),   then 
\[\rho(\ell,\ell') = 
\Pr \big(v(\ell)-v(\ell')> \epsilon(\ell')-\epsilon(\ell)\big) = F\big(\E u(\ell)- \E u(\ell')\big),
\]
where $F$ is the cumulative distribution function of the differences in errors, $\epsilon(\ell')-\epsilon(\ell)$. Such a random choice function is termed an \df{i.i.d. additive random expected utility} (iAREU).

\subsubsection{Random Expected Utility}\label{sec:REU}

In contrast to the model of additive random utility, the \df{random expected utility} model takes the function $u$ to be random. For example, if the utility of a monetary payoff $x$ is $u(x)+\epsilon(x)$, where $\epsilon:X\to\Re$ is a random function, then the utility of a simple lottery $(x,p)$ is $v(\ell)=p(u(x)+\epsilon(x))$. The stochastic choice between two simple lotteries $\ell=(x,p)$ and $\ell'=(x',p')$ is 
\[
\rho(\ell,\ell') = \Pr\big(p[u(x)+\epsilon(x)] > p'[u(x')+\epsilon(x')] \big).
\]
The random expected utility model is axiomatized by \cite{gul2006random}. The key axiom behind the model, termed Linearity, is motivated as a stochastic version of the Independence Axiom of the expected utility theory. We shall see that Linearity is intimately tied to the weak paired choice test.

\subsubsection{Random expected utility vs.\ iAREU}\label{sec:REUvsiAREU}

For our purposes, it is important to understand the choice between iAREU and random expected utility as modeling devices. \cite{strzalecki2025stochastic} provides a detailed discussion (see Section 4.6), which we proceed to summarize. The bottom line is that the iAREU model is problematic. 

First, under the model of random expected utility, the random utility has the expected utility form with probability one. Under the iAREU model, \emph{this occurs with probability zero.} In other words, the iAREU model of stochastic choice puts zero weight on a preference over lotteries that has the expected utility form. 

Second, the iAREU model can violate monotonicity in the sense of first-order stochastic dominance. Random expected utility, in contrast, always respects first-order stochastic dominance. 

Third, random expected utility will preserve the preferences over risk of the underlying von Neumann Morgenstern utility; but iAREU can subvert these. For example, it may reverse the comparison of the risk attitudes of two decision makers (see \cite{WILCOX201189} and \cite{apesteguia2018monotone}).

\begin{comment}

\subsubsection{Random probability and probability weights}\label{sec:randomweights}
Finally, we consider models in the spirit of prospect theory (\cite{kahnemann1979prospect}). The utility of a lottery that pays $x$ with probability $p$ and $0$ with probability $1-p$ is then $v(p)u(x)$ when we normalize $u(0)=0$. A random version of prospect theory takes $\tilde v(\cdot)$ and $\tilde u(\cdot)$ to be random functions, and thus, we obtain a random utility function $\tilde v(p)\tilde u(x)$ which models random deviations from expected utility theory. We focus on random functions such that the random utility $\tilde v(p)\tilde u(x)$ is centered around a deterministic expected utility $p\,u(x)$ in symmetric fashion (for a deterministic vNM function $u$) so that it deviates from expected utility theory without favoring either the common ratio effect or the reverse common ratio effect (see Sections~\ref{sec:strongpairedtest},~\ref{sec:signtest} and~\ref{sec:prelec}).

\end{comment}

\subsection{Random choice in \texorpdfstring{\cite{mcgranaghan2024distinguishing}}{McGranaghan et al}}\label{sec:mcgraassumptions}

\mnoss formulate their model using a reduced-form device. They use the iAREU model to show that weak paired choice tests are biased, but then move to a reduced-form model that is particularly well suited to analyzing valuation tests. The model is motivated as a generalization of iAREU.

We lay out their assumptions here.

\smallskip
\noindent\textbf{Assumption 1.} There is a strictly increasing function $\Gamma:\Re^2\to\Re$ with the property that $\Gamma(m, 0)=m$ for all $m$. The pair of random variables $(\ep_{\mathrm{AB}}, \ep_{\mathrm{CD}})$ is drawn from a continuous joint distribution with convex support such that
\[\rho(A, B)=\Pr(x\ge \Gamma(m^*_{\mathrm{AB}}, \ep_{\mathrm{AB}}))\]
and 
\[\rho(C, D)=\Pr(x\ge \Gamma(m^*_{\mathrm{CD}}, \ep_{\mathrm{CD}})).\]

For iAREU, it is easy to see that $\Gamma(m, \ep)=u^{-1}\big(u(m)+\ep\big)$ and $m^*_{\mathrm{AB}}=m^*_{\mathrm{CD}}=u^{-1}(p\,u(y))$.  

It is important to note that noise terms $\ep_{\mathrm{AB}}$ and $\ep_{\mathrm{CD}}$ are fundamentally different from error terms $\epsilon(\ell)$. In particular, $\epsilon(\ell)$ captures perceptual errors and preference heterogeneity/random preferences, while $\ep_{\mathrm{AB}}$ and $\ep_{\mathrm{CD}}$ are residual terms that can be derived from $\epsilon(\ell)$ given the random choice model and von Neumann-Morgenstern utility $u$ (see Section~\ref{sec:strongpairedtest} and the model in Equation~\ref{eq:perception} for an example of how these errors may be obtained from more primitive preference shocks). 

\mnoss show that the weak paired choice test is biased under either of the following two assumptions.

\bigskip
\noindent\textbf{Assumption 2a:} $\Gamma(m, \ep)=m+\ep$, $\ep_{\mathrm{AB}}\overset{d}{=} k\,\ep_{\mathrm{CD}}$ for $k>0$, and $\E[\ep_{\mathrm{AB}}]=\E[\ep_{\mathrm{CD}}]=0$. 

\bigskip
\noindent\textbf{Assumption 2b:} $\ep_{\mathrm{AB}}\overset{d}{=} k\,\ep_{\mathrm{CD}}$ for $k>0$, and $\ep_{\mathrm{AB}}$ is symmetric around zero.\bigskip

Comparing the two sets of assumptions, note that Assumption 2a imposes stronger restrictions on $\Gamma$ but weaker restrictions on noise terms, while Assumption 2b imposes no restrictions on $\Gamma$ but stronger assumptions on noise terms. Either way, the assumptions impose that $\ep_{\mathrm{AB}}\overset{d}{=} k\,\ep_{\mathrm{CD}}$ for $k>0$ and $\E[\ep_{\mathrm{AB}}]=\E[\ep_{\mathrm{CD}}]=0$. When introducing Assumption 2a, \mnoss interpret the noise terms as an additive disturbance to an underlying value. They think of the additive model as a statistical model of  valuations, not as the result of an underlying model of stochastic choice.

\mnoss also show that the sign valuation test is unbiased under the following assumption.

\medskip
\noindent\textbf{Assumption 3:} The joint distribution $(\ep_{\mathrm{AB}}, \ep_{\mathrm{CD}})$ is symmetric around some median vector $(\varepsilon', \varepsilon')$.\medskip

Assumption 3, when compared to Assumption 2b, drops the requirement $\ep_{\mathrm{AB}}\overset{d}{=} k\,\ep_{\mathrm{CD}}$ for $k>0$, but imposes joint symmetry (a central symmetry of the joint density function). Assumption 3 implies that each marginal distribution is also symmetric around $\ep'$; and the symmetry around a median $\ep'$ implies that the mean is also equal to $\ep'$.\footnote{We are assuming here that the errors have finite means.} 

Thus, all three assumptions in \mnoss imply that $\E[\ep_{\mathrm{AB}}] = \E[\ep_{\mathrm{CD}}]$. We discuss this property in Proposition~\ref{prop:pvtestbiased} and Proposition~\ref{prop:epABepCD} below.

\begin{remark}\label{rmk:assm32b}
    Observe that Assumption 3 is not weaker than 2b.  For example, suppose $(\epsilon_{AB}, \epsilon_{CD})=(X, X)$ with probability $0.8$ and $(\epsilon_{AB}, \epsilon_{CD})=(X, -X)$ with probability $0.2$ where $X \sim \text{Uniform}[-1, 1]$. Note that $\ep_{\mathrm{AB}}\sim \text{Uniform}[-1, 1]$ and $\ep_{\mathrm{CD}}\sim \text{Uniform}[-1, 1]$ and $\ep_{\mathrm{AB}}\overset{d}{=} \ep_{\mathrm{CD}}$. Yet, $(\ep_{\mathrm{AB}}, \ep_{\mathrm{CD}})$ is not symmetric. (In Appendix~\ref{sec:assmn2b3} we provide a more elaborate example that has a positive density on $[-1,1]^2$.)
\end{remark}

\subsection{Tasks}\label{sec:tasks}

Given two lotteries, $\ell$ and $\ell'$, a \df{paired choice task} elicits a choice between $\ell$ and $\ell'$.  The common ratio effect typically involves two paired-choice tasks:

\smallskip
\noindent \textbf{AB choice:}
\[A=(x, 1)\text{ or } B=(y, p).\]
\textbf{CD choice:}
\[C=(x, r)\text{ or } D=(y, r\,p).\]

A \df{valuation task} takes as given a probability $q$ and a lottery $\ell$. Then it elicits a monetary quantity $m$ such that $(m,q)$ is indifferent to $\ell$. \mnoss use \df{paired valuation tasks} motivated by the AB and CD choices in the common ratio effect. This task elicits two monetary amounts $m_{\mathrm{AB}}$ and $m_{\mathrm{CD}}$ such that 
\begin{align*}
    (m_{\mathrm{AB}},1) & \sim (y,p) \\
    (m_{\mathrm{CD}},r) & \sim (y,rp)
\end{align*}

\subsection{Paired Choice Tests}\label{sec:tests}

The linearity axiom of \cite{gul2006random}, a stochastic-choice version of the independence axiom, implies that $\rho(A, B)=\rho(C, D)$. The \df{weak paired choice test} for expected utility theory is that $\rho(A, B)=\rho(C, D)$. When this is rejected in favor of $\rho(A, B)>\rho(C, D)$, we say that the stochastic choice $\rho$ exhibits the \df{common ratio effect}.

The \df{strong paired choice test} for expected utility theory is that $\rho(A, B)\geq 1/2$ if and only if $\rho(C, D)\geq 1/2$. When this is rejected in favor of $\rho(A, B)\ge 1/2$ and $\rho(C, D)<1/2$ (or $\rho(A, B)>1/2$ and $\rho(C, D)\le 1/2$), we say that the stochastic choice $\rho$ exhibits the \df{common ratio effect}. \cite{kahnemann1979prospect} described the common ratio effect in these terms, as a rejection of the strong paired choice test. Indeed, the common ratio effect is defined in this way by \cite{ballinger1997}, who pioneered the literature on testing the common ratio effect in a stochastic choice framework (see p.1092).

The idea behind the strong test goes back to a classical question in welfare economics and stochastic choice: when can we conclude that $A$ is preferred to $B$? If one insists on deterministic choice, we might decide that $A$ is preferred to $B$ only when $\rho(A, B)=1$. For empirical purposes, however, one has to work with stochastic choice.  We may then say that $A$ is preferred to $B$ whenever $\rho(A, B)$ is large enough (e.g., see \cite{fishburn1978choice}); and it is natural to say that large enough means $\rho(A, B)\ge \rho(B,A)$. So we might decide that $A$ is preferred to $B$ when $\rho(A,B)\geq 1/2$. This is in line with  Logit models of discrete choice (in fact, by all Fechnerian models), and used by \cite{kahnemann1979prospect}'s descriptions of the common ratio effect as well as the common consequence effect. 

Experimental evidence on the Allais paradox often involves a between-subjects analysis: Stochastic choice reflects the choices of a population of individuals, and the data are only useful to the extent that the population shares some common preferences. Here we seek to infer whether $A$ is commonly preferred to $B$ from observing $\rho(A,B)$. And again, $\rho(A,B)\geq 1/2$ is a natural requirement for the inference that $A$ is ranked above $B$.

What does this mean for the common ratio effect and the Allais paradox? In Allais' original deterministic-choice thought experiment, he would require $A$ to be preferred over $B$ and $D$ to be preferred over $C$ (with at least one strict preference). Given our discussion of how to formulate the question in a setting of stochastic choice, the minimum requirement for the common ratio effect is that $\rho(A, B)>1/2\ge \rho(C, D)$ (or $\rho(A, B)\ge 1/2>\rho(C, D)$). See \cite{ballinger1997} for a similar argument. Note that, while these inequalities are a minimal requirement for the effect, they imply that $\rho(A, B)>\rho(C, D)$. Hence, the strong test requires stronger evidence than the weak test does before it rejects expected utility theory.\footnote{\cite{ballinger1997} note that the inequality, $\rho(A, B)\ge 1/2\ge \rho(C, D)$ with one strict inequality, is necessary for detecting the common ratio effect, but additional assumptions are required for sufficiency. Our theoretical results in Section~\ref{sec:strongpairedtest} clarify the assumptions (or stochastic choice models) required for sufficiency.} 

\subsection{Valuation Tests}\label{sec:valtests}

In contrast to the paired choice tests, the valuation test involves eliciting valuations for lotteries $B$ and $D$ as described above. The \df{valuation test} is the statement that $m_{\mathrm{AB}}=m_{\mathrm{CD}}$.

When choice is stochastic, these valuations will be random, even when the decision maker is consistent with expected utility theory. Therefore, we have two versions of the valuation test:
\begin{itemize}
    \item The \df{sign test} requires that $\Pr(m_{\mathrm{AB}}>m_{\mathrm{CD}})=1/2$.
    \item The \df{mean test} requires that $\E[m_{\mathrm{AB}}]=\E[m_{\mathrm{CD}}]$.
\end{itemize}
When the sign test is rejected in favor of $\Pr(m_{\mathrm{CD}}>m_{\mathrm{AB}})>1/2$, or the mean test is rejected in favor of $\E[m_{\mathrm{CD}}]>\E[m_{\mathrm{AB}}]$, then we say that the decision-maker exhibits the common ratio effect.

We start by pointing out that valuation tests do not provide sharp predictions under the iAREU model that motivated the move away from the weak paired choice test. One might have hoped that valuation tests were more robust to individual heterogeneity and noise, but the following proposition shows that, in a sense, ``anything goes'' for the valuation tests. In the next proposition, we assume that $\Gamma(m,\ep) = u^{-1}(u(m)+\ep)$ for some CRRA von-Neumann-Morgenstern utility function $u$.

\clearpage

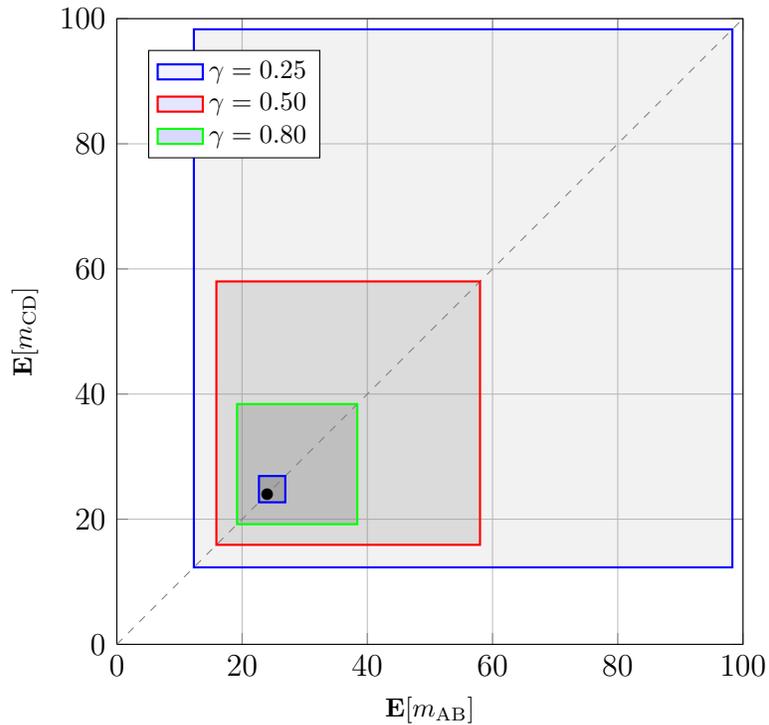
\begin{figure}[t]
    \centering
\begin{center}
\begin{tikzpicture}{scale=.5}
\begin{axis}[
    width=0.65\textwidth,
    height=0.65\textwidth,
    xlabel={\footnotesize $\E[m_{\mathrm{AB}}]$},
    ylabel={\footnotesize $\E[m_{\mathrm{CD}}]$},
    xmin=0, xmax=100,
    ymin=0, ymax=100,
    grid=both,
    legend style={at={(0.05,0.95)}, anchor=north west},
]

% 45-degree line to show symmetry
\addplot [dashed, color=gray, domain=0:100, forget plot] {x};
%\addlegendentry{Symmetry Line}

% \gamma = 0.25 (E_min = 12.3, E_max = 98.3)
\draw [fill=black!50, fill opacity=0.1, draw=blue, thick] 
    (12.3, 12.3) rectangle (98.3, 98.3);
\addlegendimage{area legend, fill=blue!50, fill opacity=0.1, draw=blue, thick}
\addlegendentry{\footnotesize $\gamma = 0.25$}

% \gamma = 0.35 (E_min = 15.9, E_max = 58.0)
\draw [fill=black!50, fill opacity=0.2, draw=red, thick] 
    (15.9, 15.9) rectangle (58.0, 58.0);
\addlegendimage{area legend, fill=blue!50, fill opacity=0.2, draw=red, thick}
%\addlegendentry{$\gamma = 0.35$}

% \gamma = 0.50 (E_min = 19.2, E_max = 38.4)
\draw [fill=black!50, fill opacity=0.3, draw=green, thick] 
    (19.2, 19.2) rectangle (38.4, 38.4);
\addlegendimage{area legend, fill=blue!50, fill opacity=0.3, draw=green, thick}
\addlegendentry{\footnotesize $\gamma = 0.50$}

% \gamma = 0.80 (E_min = 22.7, E_max = 26.9)
\draw [fill=black!50, fill opacity=0.4, draw=blue, thick] 
    (22.7, 22.7) rectangle (26.9, 26.9);
\addlegendimage{area legend, fill=blue!50, fill opacity=0.4, draw=blue, thick}
\addlegendentry{\footnotesize $\gamma = 0.80$}

%% \gamma = 1.0 (Point at 24, 24)
\addplot [only marks, mark=*, color=black, mark size=2pt] coordinates {(24, 24)};
%\addlegendentry{$\gamma = 1.0$}

\end{axis}
\end{tikzpicture}
\end{center}
    \caption{Possible values of $(\E[m_{\mathrm{AB}}],\E[m_{\mathrm{CD}}])$ as a function of $\gamma$ when $y=30$, $p=0.8$ and $r=0.4$. When $\gamma=1$ we obtain the point $(py,py) = (24,24)$. As $\gamma\to 0$ we obtain any vector in $\Re^2_{++}$.}
    \label{fig:anythinggoes}
\end{figure}

\begin{proposition}\label{prop:anythingoes} Consider an expected utility agent with a CRRA von-Neumann-Morgenstern utility function $u(x)=x^{\gamma}$, where $\gamma\in (0,1)$. Fix $y>0$ and $p,r\in (0,1)$. Suppose that $p>1/2$.
  \begin{enumerate}
    \item For any $(z_1,z_2)\in\Re^2_{++}$, there exists $\gamma\in (0,1)$ and a distribution of $(\ep_{\mathrm{AB}},\ep_{\mathrm{CD}})$ satisfying Assumption 2b, and for which each marginal distribution is symmetrically distributed around its mean and median (which equals 0), such that $$(z_1,z_2)=(\E[m_{\mathrm{AB}}],\E [m_{\mathrm{CD}}]).$$
    \item For any $q\in (0,1)$, there exists a distribution of $(\ep_{\mathrm{AB}},\ep_{\mathrm{CD}})$ satisfying Assumption 2b, and for which each marginal distribution is symmetrically distributed around its mean and median (which equals 0), for which \[ q = \Pr(m_{\mathrm{AB}}>m_{\mathrm{CD}}).\]
      \end{enumerate}
\end{proposition}

The first statement of Proposition~\ref{prop:anythingoes} is illustrated in Figure~\ref{fig:anythinggoes}, which shows the range of possible values of the expected valuations as the coefficient of relative risk aversion $\gamma$ decreases to $0$.

\begin{remark}
Proposition~\ref{prop:anythingoes} is in contrast with Proposition 2 of \mnossc, which shows that $\E[m_{\mathrm{AB}}] = \E[m_{\mathrm{CD}}]$ when $\Gamma(m,\ep)=m+\ep$, and that $\Pr(m_{\mathrm{AB}}> m_{\mathrm{CD}})=1/2$ under Assumption 3.
    
 The explanation for this discrepancy is that the mean test is very sensitive to the curvature of $\Gamma$, which depends on the agents' level of risk aversion. In the context of the iAREU model, Proposition 2 of \mnoss requires risk neutrality. We show that the level of risk aversion may be chosen to achieve any pair of non-negative expected valuations. The problems with the mean tests are acknowledged by \mnossc, who show that it can be biased under Assumption 2b.\footnote{This is shown formally in their online appendix and mentioned in the main text of the paper. Their main message is that there are assumptions under which the traditional weak paired choice test may be problematic, while tests using paired valuations may not be problematic.}

The sign test, on the other hand, is sensitive to the correlation between errors. When errors are independent, of course, we obtain that $\Pr(m_{\mathrm{AB}}> m_{\mathrm{CD}})=1/2$. But it seems important to allow for deviations from independence, especially (but not exclusively) if the test is to be used for individual-level data. Assumption 3 allows for correlated errors and ensures that the sign test is unbiased; however, it also limits the correlation among errors in significant ways, as we saw in the remark following Assumption 3 above. 

The issues with the sign test are not only about correlation: A version of statement 2 in Proposition~\ref{prop:anythingoes} holds for independent errors, see Appendix~\ref{sec:pairedval}.
\end{remark}

\section{Paired choice tests}\label{sec:choicetests}

In light of the issues with valuation tests highlighted by Proposition~\ref{prop:anythingoes}, we proceed to discuss the paired choice tests. We examine paired choice tests in the context of the different models of stochastic choice that we have introduced.

\subsection{The iAREU model}\label{sec:pairedadditive}

Following \cite{Loomes_2005} and \cite{mcgranaghan2024distinguishing},  we first consider the weak paired choice test under the Fechnerian iAREU model of stochastic choice (see Section~\ref{sec:iAREU}). Let $u:\Re_+\to\Re$ be a vNM utility and $G:\Re\to[0, 1]$ be strictly increasing and symmetric such that  $$\rho(\ell_A, \ell_B)=G\big(\E[u(\ell_A)]-\E[u(\ell_B)]\big).$$

Then, for the weak paired choice test, we have 
\begin{align*}
  \rho(A, B) & =G\big(u(x)-p\,u(y)\big) \\  
  \rho(C, D) & =G\big(r\,(u(x)-p\,u(y))\big).
\end{align*}

Hence,
\[\rho(A, B)=G\big(u(x)-p\,u(y)\big)\neq G\big(r\,(u(x)-p\,u(y))\big)=\rho(C, D)\] 
unless $p\,u(y)=u(x)$. More importantly,
\[\rho(A, B)>\rho(C, D)\text{ if }u(x)-p\,u(y)>0.\]

So the weak paired choice test is biased in the sense that it detects the common ratio effect even when the underlying model is iAREU, and therefore ``consistent'' with expected utility. In making this point, \cite{Loomes_2005} cautioned that the statement ``$\rho(A, B)$ is significantly higher than $\rho(C, D)$'' is sensitive to the assumptions made regarding the underlying model of stochastic choice. Under the iAREU model, the common ratio effect, as defined through the weak paired choice test, should not be seen as grounds to reject expected utility theory. \mnoss use the bias of the weak paired choice under the iAREU as a starting point to motivate their focus on valuation tests.

The status of iAREU as the incarnation of expected utility in stochastic choice is, however, questionable. There are several arguments against the iAREU: the iAREU model delivers non-expected utility preferences with probability one; it may violate monotonicity with respect to first-order stochastic dominance, and it can overturn the decision-maker's underlying risk preferences. See our discussion in Section~\ref{sec:REUvsiAREU}. What should we then make of the bias of the weak paired choice test? Is the bias of the test due to problems with the iAREU as a stochastic incarnation of expected utility theory; or is the test itself problematic?

\subsection{Random Expected Utility Model}\label{sec:pairedtestREU}

In random expected utility (Section~\ref{sec:REU}), the random utility of a binary lottery $(x, p)$ is given by $U(x, p)=p (u(x)+\epsilon_x)$, where $\epsilon_x$ captures a utility shock. The paired choice test is unbiased under this model of stochastic choice because 
\begin{align*}
  \rho(A, B)&=\Pr\big(u(x)+\epsilon_x>p(u(y)+\epsilon_y)\big) \\
  &=\Pr\big(r\,(u(x)+\epsilon_x)>r\, p(u(y)+\epsilon_y)\big)
=\rho(C, D).  
\end{align*}

In fact, the unbiasedness of the test holds without needing to make any assumptions about the (joint) distribution of errors $(\epsilon_x,\epsilon_y)$.

Notice that the random expected utility is a special case of additive random utility where $\epsilon_{p, x}=p\,\epsilon_x$. However,  the i.i.d.\ assumption made in the iAREU model~\eqref{sec:pairedadditive} is violated under random expected utility. Under random expected utility, errors $\epsilon_{p, x}=p\,\epsilon_x$ are independent but not identically distributed when $\epsilon_x$ are i.i.d. 

\cite{gul2006random} propose a behavioral justification for random expected utility in the form of the stochastic analogue of the independence axiom of classical expected utility theory. Independence requires that $\ell$ is preferred to $\ell'$ if and only if $\lambda \ell + (1-\lambda) \ell''$ is preferred to $\lambda \ell' + (1-\lambda) \ell''$. The version of the independence axiom in the context of stochastic choice, called Linearity, is:
\[ 
\rho(\ell, \ell')=\rho(\lambda \ell + (1-\lambda) \ell'', \lambda \ell' + (1-\lambda) \ell'').\footnote{This is Linearity for doubleton menus. For general menus, it is defined as: $\rho(\ell \mid D)=\rho(\lambda \ell + (1-\lambda) \ell' \mid  \lambda D + (1-\lambda) \{\ell'\})$.}
\]
Linearity is the key axiom that characterizes random expected utility, and justifies this model as the ``correct'' extension of expected utility theory to the stochastic choice framework. \cite{gul2006random} write ``\emph{Studies that investigate the empirical validity of expected utility theory predominantly use a random choice setting. For example, Kahneman and Tversky describe studies that report frequency distributions of choices among lotteries. These studies test expected utility theory by checking if the choice frequencies remain unchanged when each alternative is combined with some ﬁxed lottery; that is, by testing our linearity axiom.}''\footnote{They also write ``Linearity is analogous to the independence axiom of the von Neumann--Morgenstern theory. Note that this ``version'' of the independence axiom corresponds exactly to the version used in experimental settings. In the experimental setting, a group of subjects is asked to make a choice from a binary decision problem $D = \{x, x'\}$. Then the same group chooses from a second decision problem obtained from the first by replacing the lotteries $x \succsim x'$ in the original problem with $\lambda x + (1-\lambda)y$ and $\lambda x' + (1-\lambda)y$. Linearity requires that the frequency with which the lottery $x$ is chosen in the first problem is the same as the frequency with which the lottery $\lambda x + (1-\lambda)y$ is chosen in the second problem."} Their arguments imply that if one assumes a stochastic choice model that violates Linearity, then one already assumes, from the beginning, that a basic tenet of the theory is violated. We formalize this as Proposition~\ref{prop:linearity}.

Indeed, the iAREU model violates Linearity. To illustrate, consider lotteries $\ell_1=(\$10, 1)$ and $\ell_2=(\$30, 0.9)$. The random utility of $(x, p)$ is $U(x, p)=p\, u(x)+\epsilon$ where $\epsilon\sim U[-1, 1]$ is i.i.d. Let $u(10)=8$ and $u(30)=12$. Hence, $\rho(\ell_1, \ell_2)=0$. However, when $\ell_3=(0, 1)$ and $\lambda=0.25$, we have $\lambda \ell_1 + (1-\lambda) \ell_3=(\$10, 0.25)$ and $\lambda \ell_2 + (1-\lambda)\ell_3=(\$30, 0.225)$. Hence, 
\[ 
\rho(\lambda \ell_1 + (1-\lambda) \ell_3, \lambda \ell_2 + (1-\lambda) \ell_3)=\frac{169}{800}
>\rho(\ell_1, \ell_2)=0.
\]

This example shows that the iAREU model exhibits counterintuitive behavior in contexts that are similar to the Allais paradox. However, the problem is deeper than a single counterexample may suggest. Consider the following version of Linearity, formulated for simple binary lotteries. This version essentially amounts to the weak paired choice test. 

\smallskip
\noindent\textbf{Weak Linearity:} For any $(x, r)$ and $(y, q)$, 
\[\rho((x, 1), (y, q))=\rho((x, r), (y, r\,q)).\]

It turns out that Weak Linearity is fundamentally in conflict with the iAREU. To formalize this observation, we shall need an additional assumption. A stochastic choice function $\rho$ is \textbf{scalable} with respect to expected utility if there are functions $F:\Re^2_+\to [0, 1]$ and $u:\Re_{+}\to\Re_{+}$ such that $F$ is  strictly increasing in the first argument when $F\in (0, 1)$, and for any $\ell_A, \ell_B$,
\[\rho(\ell_A, \ell_B)=F\big(\E[u(\ell_A)], \E[u(\ell_B)]\big).\]

\begin{prps}\label{prop:linearity} Any stochastic choice function that is scalable with respect to expected utility violates Linearity. If it satisfies Weak Linearity, then it is Fechnerian but not iAREU.
\end{prps}

The proof of Proposition~\ref{prop:linearity} is in Section~\ref{sec:proofs}, as are all proofs in the paper. 

We now turn to the connection between the strong paired choice test, valuation tests, and Linearity. We may relax Weak Linearity in the following two ways. 

\medskip
\noindent\textbf{Tied Linearity I:} For any $(m, r)$ and $(y, q)$, 
\[\rho((m, 1), (y, q))=\frac{1}{2}\text{ iff }\rho((m, r), (y, r\,q))=\frac{1}{2}.\]

\medskip
\noindent\textbf{Tied Linearity II:} For any $(x, r)$ and $(y, q)$, 
\[\rho((x, 1), (y, q))\ge \frac{1}{2}\text{ iff }\rho((x, r), (y, r\,q))\ge \frac{1}{2}.\]

Tied Linearity II implies Tied Linearity I. The reverse is true under a monotonicity assumption that is satisfied throughout the paper. Hence, Tied Linearity I and II are equivalent for our purposes. Monotonicity is also assumed by the multiple-price list method used to elicit valuations. 

The strong paired choice test assesses Tied Linearity II, while the paired valuation test essentially evaluates Tied Linearity I. Thus, at a conceptual level, the strong paired choice and valuation tests are assessing the same weakening of Weak Linearity. However, a difference is that the nature of the valuation test (where certainty-equivalents are elicited) requires some additional assumptions. This means that the strong paired choice test is unbiased under weaker assumptions than the valuation test.

\subsection{The unbiasedness of the strong paired choice test}\label{sec:strongpairedtest}

In this section, we show that the strong paired choice test is robustly unbiased across various standard models of stochastic choice.

The strong paired choice test is unbiased under the Fechnerian model of iAREU. The statement that $\rho(A, B)\ge \frac{1}{2}$ is equivalent to $F\big(u(x)-p\,u(y)\big)\ge F(0)=\frac{1}{2}$, while  $\rho(C, D)\ge \frac{1}{2}$ is equivalent to $F\big(r(u(x)-p\,u(y))\big)\ge F(0)=\frac{1}{2}$. Hence, 
\[\begin{split}
\rho(A, B)\ge \frac{1}{2}\text{ if and only if }u(x)-p\,u(y)\ge 0 \\ \text{ if and only if  }\rho(C, D)\ge \frac{1}{2}.    
\end{split}
\]

More formally, the strong paired choice test is unbiased as long as the stochastic choice satisfies \emph{weak stochastic transitivity} with respect to expected utility: there is $u:\Re_+\to\Re$ such that 
\[\rho(\ell_A, \ell_B)\ge \frac{1}{2}\text{ if and only if }\E u(\ell_A)\ge \E u(\ell_B).\] In light of our discussion of how preferences are inferred from stochastic choices, the connection between $\rho(\ell_A, \ell_B)\ge 1/2$ and $\E u(\ell_A)\ge \E u(\ell_B)$ is natural. In particular, Fechnerian models satisfy strong stochastic transitivity with respect to expected utility. 

Consider a generalization of the Fechnerian model that is motivated by \cite{he2024moderate}. The stochastic choice $\rho$ has a \df{weak expected utility representation} if there is a cdf $F$, and a strictly positive valued function $G:\Re\to\Re_{++}$, such that $\rho(\ell_A, \ell_B)=F\big(G(\ell_A, \ell_B)\big[\E[u(\ell_A)]-\E[u(\ell_B)]\big]\big)$.\footnote{\cite{he2024moderate} characterizes the weak utility presentation $\rho(\ell_A, \ell_B)=F\big(G(\ell_A, \ell_B)\big[v(\ell_A)-v(\ell_B)\big]\big)$ by means of weak stochastic transitivity. The above model is the expected utility version of the weak utility representation.}

\begin{prps}
\label{prop:weakEU}
Suppose that the stochastic choice $\rho$ has a \df{weak expected utility representation}. Then
\[\rho(A, B)\ge \frac{1}{2}\text{ if and only if }\rho(C, D)\ge \frac{1}{2}.\]
\end{prps}

Many stochastic choice models have a weak utility representation when errors are correlated. For example, some of the best known discrete choice models in economics and psychology such as the covariance probit (\cite{thurstone1927psychophysical}), nested logit (\cite{ben1973structure}, \cite{mcfadden1978modeling}), the elimination-by-aspects (\cite{tversky1972elimination}), and the random coefficients model of \cite{hausman1978conditional} have a weak utility representation, but not a Fechnerian one.\footnote{See \cite{wilcox2008stochastic} and \cite{he2024moderate} for more examples.} Hence, the expected utility version of these models will have weak expected utility representations. It is also important to note that weak utility representation goes beyond the random utility framework. Hence, the above result demonstrates that the strong test is unbiased under a general class of stochastic choice models.

Finally, we turn to the strong paired choice under the assumptions of \mnoss (see Section~\ref{sec:mcgraassumptions}). \mnoss show that the weak paired choice test is biased under either Assumption 2a or 2b, while their sign valuation test is unbiased under Assumption 3. We first show that the strong paired choice test is unbiased under an assumption weaker than both Assumption 2b and Assumption 3. In Appendix~\ref{sec:meantestbias}, we show that their Assumption 2a is often violated (a fact already acknowledged by \mnoss). 

\begin{prps}
\label{prop:assumption2} Assume Assumption 1 and that $m^*_{\mathrm{AB}}=m^*_{\mathrm{CD}}$. Suppose that $\Pr(\ep_{\mathrm{AB}}\le \ep')=\Pr(\ep_{\mathrm{CD}}\le \ep')=\frac{1}{2}$ for some constant $\ep'$. Then,
\[\rho(A, B)\ge \frac{1}{2}\text{ if and only if }\rho(C, D)\ge \frac{1}{2}.\]
\end{prps}

The assumption on noises in Proposition~\ref{prop:assumption2} relaxes both Assumption 2b and Assumption 3. First, we drop the assumption $\ep_{\mathrm{AB}}\overset{d}{=} k\,\ep_{\mathrm{CD}}$ for $k>0$ from Assumption 2b. Second, $\Pr(\ep_{\mathrm{AB}}\le \ep')=\Pr(\ep_{\mathrm{CD}}\le \ep')=\frac{1}{2}$ holds whenever $\ep_{\mathrm{AB}}$ and $\ep_{\mathrm{CD}}$ are symmetric around some constant $\ep'$, which is required by both Assumption 2b and Assumption 3.

Assumption 2b is satisfied under the iAREU model, which has a Fechnerian representation. Below, we show that guaranteeing either $\ep_{\mathrm{AB}}\overset{d}{=} k\,\ep_{\mathrm{CD}}$ or the assumption $\E[\ep_{\mathrm{AB}}]=\E[\ep_{\mathrm{CD}}]$ is challenging once we deviate from iAREU. 

To demonstrate that the strong paired choice test can be unbiased beyond Assumption 2b and Assumption 3, consider the following model that allows for a ``perceptual'' error (or random probability weighting) and a utility over money.
\begin{equation}\label{eq:perception}
U(x, p)=p\,u(x)+\epsilon_p\,f(x)+\epsilon_x\, g(p)+\epsilon_{p, x}.
\end{equation}
Here,  $\epsilon_p$ captures a perceptual error, $\epsilon_x$ captures randomness in the evaluation of monetary payoffs,  $\epsilon_{p, x}$ captures the remaining error (or the interaction of perceptual and preference heterogeneity). Such a model obtains, for example, when the utility of $(x,p)$ is $(p+\epsilon_p)[u(x)+\epsilon_x]$.\footnote{Note that this model nests both cases where $\epsilon_1=0$ or $\epsilon_1$ and $\epsilon_p$ with $p<1$ are identical distributions.}

\begin{prps}\label{prop:stpairchoiceab} Consider the model in Equation~\eqref{eq:perception} and suppose that one of the following assumptions is satisfied:
\begin{enumerate}
\item All errors are symmetric about zero, and any two error terms are either independent or linearly dependent. 
\item Let $\epsilon_{p, x}=\alpha\,\epsilon_p\,\epsilon_x+\beta\,\epsilon^*_{p, x}+\gamma$ for some constants $\alpha, \beta, \gamma$ where errors $\epsilon^*_{p, x}$ are identical or symmetric about zero. All other errors are independent and symmetric about zero.
\item Errors $\epsilon_{p, x}$ are independent, either identical or symmetric around the same fixed constant, and independent from $\epsilon_p$ and $\epsilon_x$. All other errors are symmetric about zero, and any two of them are either independent or linearly dependent.
\end{enumerate}
Then 
\[\rho(A, B)\ge \frac{1}{2}\text{ if and only if }\rho(C, D)\ge \frac{1}{2}.\]
\end{prps}

The key assumption behind Proposition~\ref{prop:stpairchoiceab} is symmetry around zero. Due to the existence of $\epsilon_{p, x}$, it is without loss of generality in~\eqref{eq:perception} to say that $\E[\epsilon_p]=\E[\epsilon_x]$. The results above assume that $\epsilon_p$ and $\epsilon_x$ are symmetric about zero. If the errors are not symmetric about zero, the model deviates systematically from expected utility. For example, suppose $\epsilon_x=\epsilon_{p, x}=0$ and $f(x)=u(x)$. Then we have $U(x, p)=(p+\epsilon_p)\,u(x)$, where the perceived probability is $p+\epsilon_p$. If $\epsilon_p$ is not symmetric, then the model itself has a bias for or against the common ratio effect. In other words, the strong paired choice test would not reject expected utility when the underlying model is unbiased.

Going back to the model of \mnoss discussed in Section~\ref{sec:mcgraassumptions}, recall that it implies that $\E[\ep_{\mathrm{AB}}] = \E[\ep_{\mathrm{CD}}]$. We now connect this property with the model in Equation~\eqref{eq:perception}.

\begin{prps}\label{prop:pvtestbiased}  Consider the model in Equation~\eqref{eq:perception} and let $f(x)=u(x)$. For each of the three assumptions of Proposition~\ref{prop:stpairchoiceab}, there exist $\epsilon_p, \epsilon_x, \epsilon_{p, x}$ that satisfy the assumption and
\[E[\ep_{\mathrm{AB}}]\neq E[\ep_{\mathrm{CD}}].\]
\end{prps}

This result and its proof show that even when $\epsilon_p, \epsilon_x, \epsilon_{p, x}$ are independent and symmetric around zero, we may have $E[\ep_{\mathrm{AB}}]\neq E[\ep_{\mathrm{CD}}]$.

\section{Detecting the Common Ratio Effect via Paired Choice Tasks}\label{sec:empirical}

We have shown that the strong paired choice test is unbiased across a broad class of stochastic choice models. Now we take stock of these considerations, and reassess the empirical evidence for the common ratio effect in light of what we know about the weak and strong choice tests.

Figure~\ref{fig:eu-cre-rcre} lays out three distinct interpretations of the paired choice test. Panel (a) shows the weak test: points on the diagonal ($\rho(A,B)=\rho(C,D)$) line up with expected utility (EU); points below the diagonal reveal the common ratio effect (CRE); and points above it reveal the reverse common ratio effect (RCRE). Panel (b) displays the strong paired choice test that we advocate for in Section~\ref{sec:choicetests}. Panel (c) reproduces Figure 2 from \mnoss, shading in gray every pattern of choice frequencies that is compatible with expected utility and a flexible random-utility framework. Concretely, the gray region collects all frequencies that can arise when agents: 1) choose either $A$ and $C$, or the pair $B$ and $D$, with probabilities $\lambda$ and $1-\lambda$, respectively; and 2) whenever an agent intends to choose a lottery $\ell$, they actually select it with probability at least $1/2$.

\begin{figure}[t]
\centering
\begin{tikzpicture}

% ===================== FIGURE 1 =====================
\begin{scope}[xshift=-1.5cm, scale=3]

    \definecolor{lightblue}{RGB}{173,216,230}

    % shading
    \fill[lightblue] (0,0) -- (1,0) -- (1,1) -- cycle;

    % square
    \draw[thick] (0,0) rectangle (1,1);

    % 45-degree line
    \draw[very thick] (0,0) -- (1,1);

    % labels
    \node at (0.7,0.2) {\scriptsize CRE};
    \node at (0.25,0.75) {\scriptsize RCRE};
    \node[above right] at (0.7,0.4) {\scriptsize EU};
    \draw[->][very thick] (0.72, 0.5) -- (0.61, 0.59);
    % ticks
    \draw (0,0) -- (0,-0.03) node[below] {\tiny $0$};
    \draw (0.5,0) -- (0.5,-0.03) node[below] {\tiny $\frac{1}{2}$};
    \draw (1,0) -- (1,-0.03) node[below] {\tiny $1$};

    \draw (0,0) -- (-0.03,0) node[left] {\tiny $0$};
    \draw (0,0.5) -- (-0.03,0.5) node[left] {\tiny $\frac{1}{2}$};
    \draw (0,1) -- (-0.03,1) node[left] {\tiny $1$};

    \node at (0.5,-0.3) {\scriptsize $\Pr(A)$};
    \node at (-0.3,0.5) {\scriptsize $\Pr(C)$};

    \node at (0.5,-0.5) {(a) \scriptsize Weak Paired Choice};

\end{scope}

% ===================== FIGURE 2 =====================
\begin{scope}[xshift=3cm, scale=3]

    \definecolor{lightblue}{RGB}{173,216,230}
    \definecolor{lightgray}{RGB}{180,180,180}

    % shading
    \fill[lightgray] (0,0) rectangle (0.5,0.5);
    \fill[lightgray] (0.5,0.5) rectangle (1,1);
    \fill[lightblue] (0.5,0) rectangle (1,0.5);

    % square and lines
    \draw[thick] (0,0) rectangle (1,1);
    \draw (0.5,0) -- (0.5,1);
    \draw (0,0.5) -- (1,0.5);

    % labels
    \node at (0.25,0.25) {\scriptsize EU};
    \node at (0.75,0.75) {\scriptsize EU};
    \node at (0.75,0.25) {\scriptsize CRE};
    \node at (0.25,0.75) {\scriptsize RCRE};

    % ticks
    \draw (0,0) -- (0,-0.03) node[below] {\tiny $0$};
    \draw (0.5,0) -- (0.5,-0.03) node[below] {\tiny $\frac{1}{2}$};
    \draw (1,0) -- (1,-0.03) node[below] {\tiny $1$};

    \draw (0,0) -- (-0.03,0) node[left] {\tiny $0$};
    \draw (0,0.5) -- (-0.03,0.5) node[left] {\tiny $\frac{1}{2}$};
    \draw (0,1) -- (-0.03,1) node[left] {\tiny $1$};

    \node at (0.5,-0.3) {\scriptsize $\Pr(A)$};
    \node at (-0.3,0.5) {\scriptsize $\Pr(C)$};

    \node at (0.5,-0.5) {(b) \scriptsize Strong Paired Choice};

\end{scope}

% ===================== FIGURE 3 =====================
\begin{scope}[xshift=7.6cm, scale=3]

    \definecolor{lightblue}{RGB}{173,216,230}
    \definecolor{lightgray}{RGB}{180,180,180}

    % shading
    \fill[lightgray] 
        (0,0) -- (0.5,0) -- (1,0.5) -- (1,1) -- (0.5,1) -- (0,0.5) -- cycle;

    \fill[lightblue]
        (0.5,0) -- (1,0) -- (1,0.5) -- cycle;

    % square
    \draw[thick] (0,0) rectangle (1,1);

    % polygon boundaries
    \draw (0,0) -- (0.5,0) -- (1,0.5) -- (1,1) -- (0.5,1) -- (0,0.5) -- cycle;
    \draw (0.5,0) -- (1,0) -- (1,0.5) -- cycle;

    % labels
    \node at (0.5,0.5) {\scriptsize EU};
    \node at (0.82,0.14) {\scriptsize CRE};
    \node at (0.17,0.88) {\scriptsize RCRE};

    % ticks
    \draw (0,0) -- (0,-0.03) node[below] {\tiny $0$};
    \draw (0.5,0) -- (0.5,-0.03) node[below] {\tiny $\frac{1}{2}$};
    \draw (1,0) -- (1,-0.03) node[below] {\tiny $1$};

    \draw (0,0) -- (-0.03,0) node[left] {\tiny $0$};
    \draw (0,0.5) -- (-0.03,0.5) node[left] {\tiny $\frac{1}{2}$};
    \draw (0,1) -- (-0.03,1) node[left] {\tiny $1$};

    \node at (0.5,-0.3) {\scriptsize $\Pr(A)$};
    \node at (-0.3,0.5) {\scriptsize $\Pr(C)$};

    \node at (0.5,-0.5) {(c) \scriptsize \mnoss};

\end{scope}

\end{tikzpicture}
\caption{Comparison of the EU, CRE, and RCRE regions under paired choice tests.}
\label{fig:eu-cre-rcre}
\end{figure}
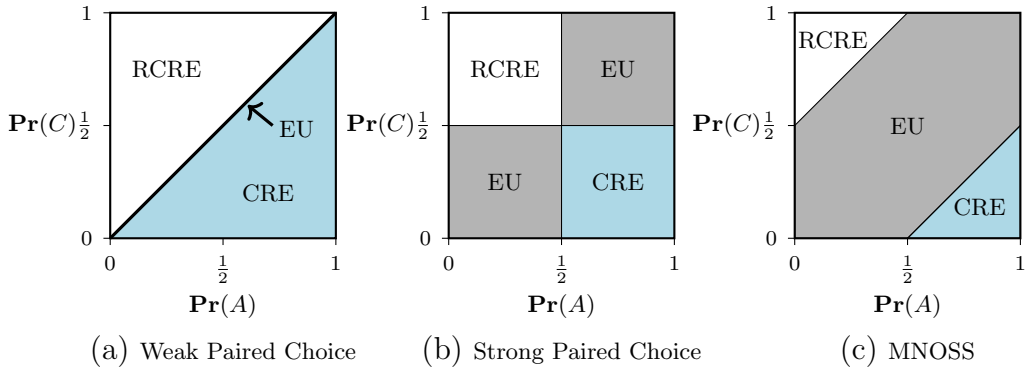

We reproduce one of \mnoss results in Figure~\ref{fig:main}, which depicts the choice frequencies in 143 studies taken from the meta-analysis in \cite{Blavatskyy_Panchenko_Ortmann_2023}. Each study is represented by a ``bubble'' that is scaled in size to reflect the number of participants in the study.\footnote{See Panel A of Figure 2 and Figure D.1 of \mnoss. Our graph is made from the data in the replication package in \cite{Blavatskyy_Panchenko_Ortmann_2023}, while Table~\ref{tab:data} is calculated from the data in \mnoss.} The choice frequencies in the literature largely fall in the gray area of Panel (c) in Figure~\ref{fig:eu-cre-rcre}. The conclusion is that the paired choice test fails to reject expected utility once one allows for general noise in participants' utilities.

In contrast, the strong test has substantial empirical bite. In Figure~\ref{fig:main}, the studies shown with purple ``bubbles'' are those that fall outside the two squares on the main diagonal. These exhibit a violation of expected utility according to the strong test. Those that fall to the right of the $0.5$ vertical line and below the $0.5$ horizontal line exhibit the common ratio effect. The studies to the left of the vertical line and above the horizontal line exhibit the reverse common ratio effect. 

\begin{figure}[t]
    \centering
\begin{subfigure}[t]{0.48\textwidth}
\includegraphics[width=2.8in]{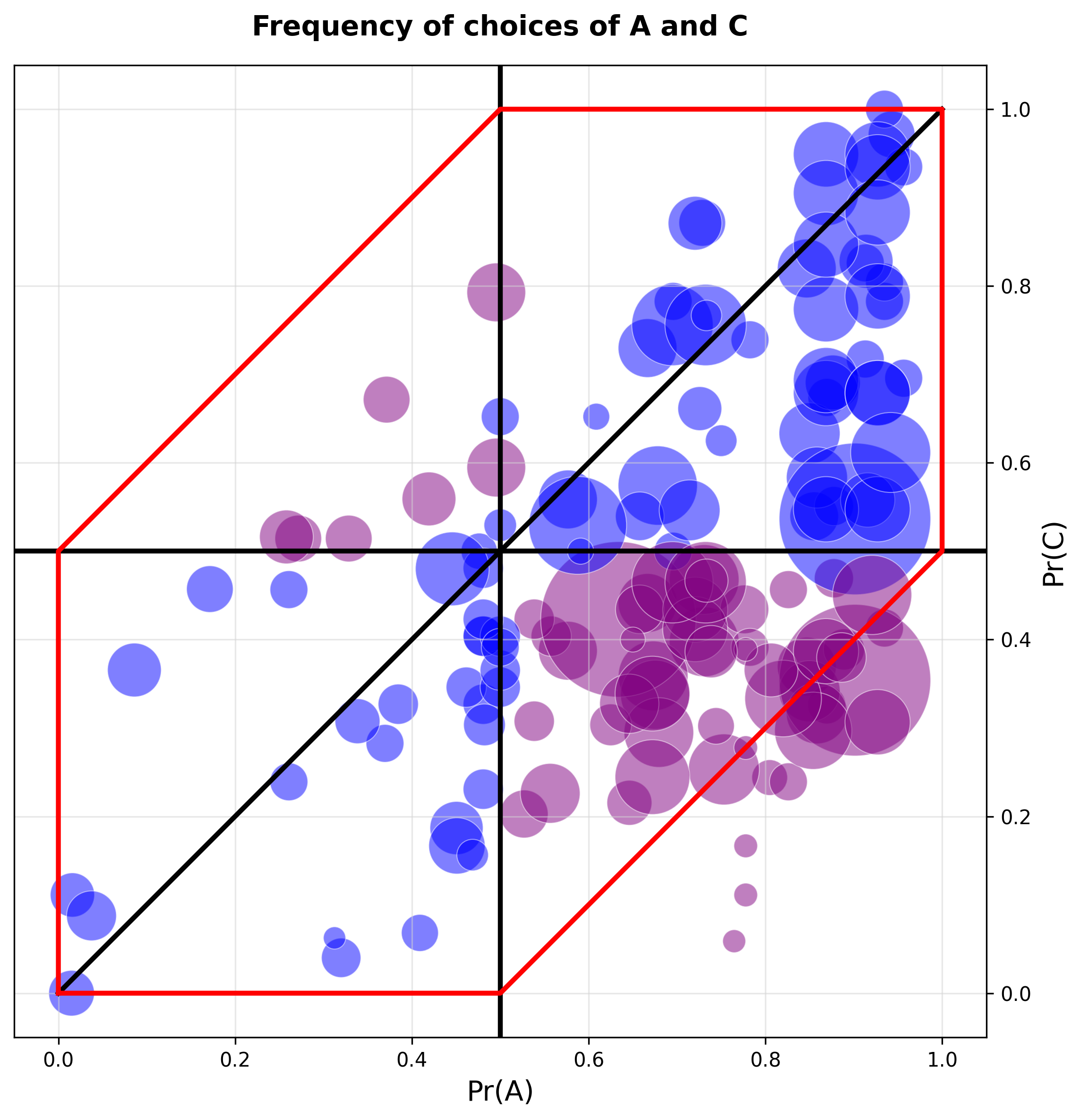}
\caption{Paired choice frequencies for the 143 experimental studies surveyed in \cite{Blavatskyy_Panchenko_Ortmann_2023}.}
\end{subfigure}\quad
\begin{subfigure}[t]{0.48\textwidth}
\includegraphics[width=2.8in]{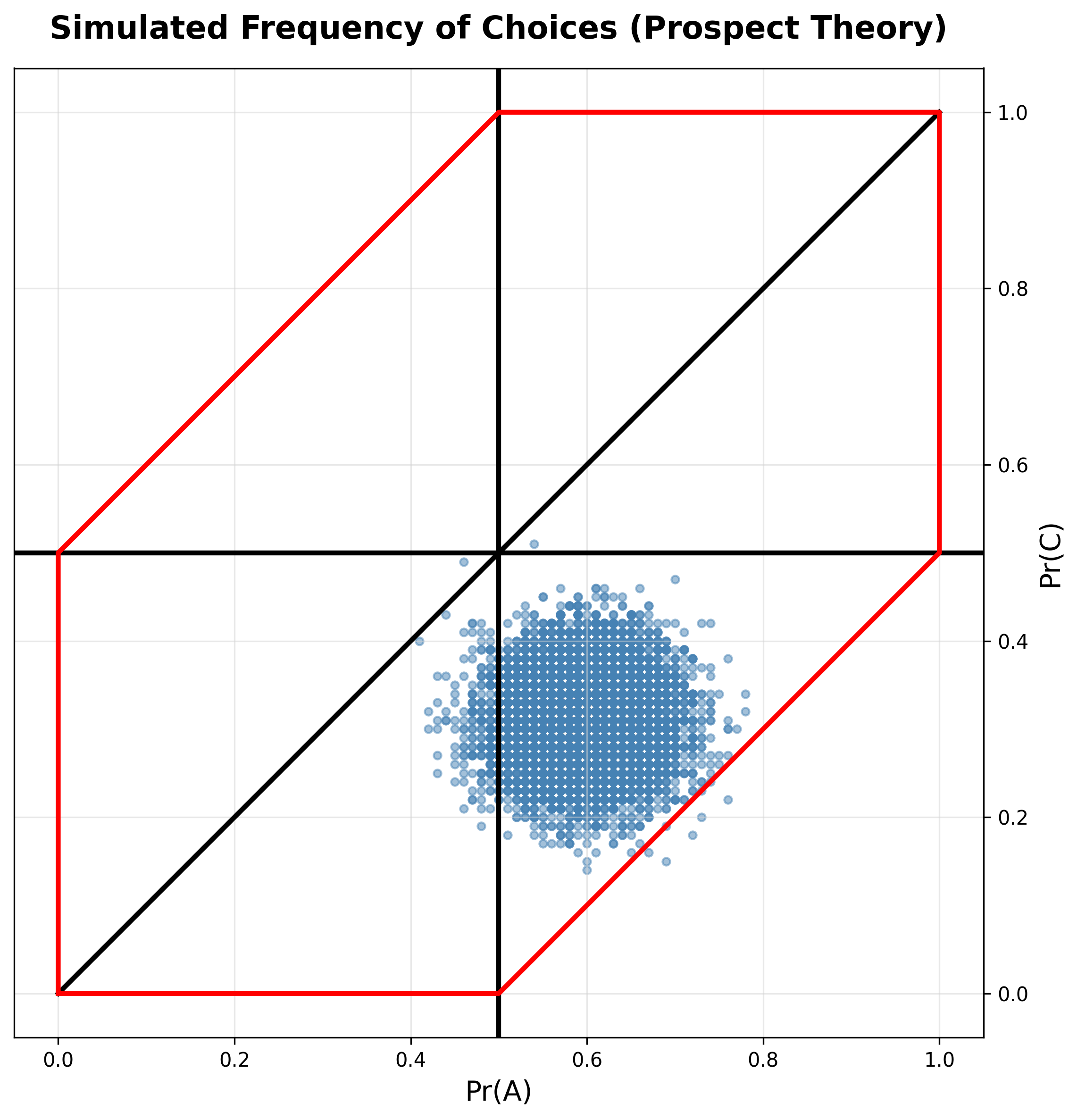}
\caption{Simulation of choices under prospect theory with  $u(x)=x^\gamma$ and distortion function $w(p)=\frac{p^\sigma}{(p^\sigma+(1-p)^\sigma)^{1/\sigma}}$.}
\end{subfigure}
\caption{Implications of Figure~\ref{fig:eu-cre-rcre} for data and simulated choices.}\label{fig:main}
\end{figure}

The discrepancy between the strong test and Panel (c) of Figure~\ref{fig:eu-cre-rcre} is easy to understand. Consider, for example, the choice frequencies $(2/3,1/3)$. In Panel (c), one could argue that these frequencies reflect an expected utility agent, or a population of agents, who would choose $A$ and $C$ with probability $2/3$; and $B$ and $D$ with probability $1/3$; however, when they are to choose $A$, they do so for sure, and when they are to choose $C$, they do so with probability $1/2$. For the strong test, we turn the $1/2$ threshold on its head: we infer from $\rho(A,B)=2/3$ that $A$ is preferred to $B$, and from $\rho(C,D)=1/3$ that $D$ is preferred to $C$. Then we conclude that these choices are inconsistent with expected utility theory.

Using the strong paired choice test, 41.26\% of the designs surveyed in \citeauthor{Blavatskyy_Panchenko_Ortmann_2023} display a common ratio effect, while 7\% display a reverse common ratio effect.\footnote{The observed percentages of the (reverse) common ratio effect in \citeauthor{Blavatskyy_Panchenko_Ortmann_2023} will be 51\% (resp., 6\%) if we allow for a 95\% confidence interval for sampling errors. For example, the observed frequencies $\hat{\rho}(A, B)=0.48$ and $\hat{\rho}(C, D)=0.3$ will be consistent with the common ratio inequality $\rho(A, B)\ge 0.5$ and $0.5>\rho(C, D)$.} The corresponding numbers in the experiments conducted by \mnoss are 10\% each. These results are summarized in Table~\ref{tab:data}. The strong test finds the same prevalence of the common ratio effect and the reverse common ratio effect in the experiments conducted by \mnossc, which is consistent with the message in their paper.  If we weight the studies in \citeauthor{Blavatskyy_Panchenko_Ortmann_2023} by the number of experimental participants, we see that almost 50\% of experimental participants across the surveyed experiments exhibit either a strong common ratio effect or a reverse common ratio effect.\footnote{The sizes of the different experiments in \mnoss are similar. There are a total of 900 participants in their experiments and 14909 in the experiments surveyed by \citeauthor{Blavatskyy_Panchenko_Ortmann_2023}.}

\begin{table}[htbp]
    \centering
    \small
    \caption{Common ratio effect in weak and strong paired tests.}
    \label{tab:data}
    \begin{tabular}{lcccccc}
        \toprule
        & & & & \multicolumn{2}{c}{\textbf{SPC}} & \\
        \cmidrule(lr){5-6}
        & \textbf{PrA} & \textbf{PrC} & \textbf{CRE} & \textbf{CRE} & \textbf{RCRE} & $n$ \\
        \midrule
        \citeauthor{mcgranaghan2024distinguishing} & 49.46\% & 46.77\% & 65.83\% & 10.00\% & 10.00\% & 120 \\
        \citeauthor{Blavatskyy_Panchenko_Ortmann_2023} & 67.37\% & 48.27\% & 79.02\% & 41.26\% & 6.99\% & 143 \\
        \bottomrule
    \end{tabular}
    
    \vspace{0.15cm}
    \parbox{0.85\textwidth}{\footnotesize \textit{Note:} CRE is the prevalence of the common ratio effect according to the weak paired test. The following two columns contain the prevalence of the common ratio effect and the reverse common ratio effect, according to the strong paired test. The number of experiments in each study is $n$.}
\end{table}

To reiterate, the criterion in Panel (c) means that expected utility is largely not rejected by the studies on the common ratio effect in the literature. The flip side of this is that the paired choice test, when interpreted as in Panel (c), has very low power to test for expected utility. To illustrate this lack of power, we simulate choices made under prospect theory.  Consider a CRRA utility function $u(x)=x^\gamma$ and distortion function $w(p)=\frac{p^\sigma}{(p^\sigma+(1-p)^\sigma)^{1/\sigma}}$. Let $\gamma=0.8$, $\sigma=0.7$, and $C=(x, r)=(12, 0.2)$ and $D=(y, rp)=(30, 0.1)$. We run 1000 simulations for $\rho(A, B)$ and $\rho(C, D)$, which are the sample averages of 100 independent choices according to the above prospect theory preference with i.i.d.\ and uniform errors $\epsilon_{AB}$ and $\epsilon_{CD}$ on $[-1.8, 1.8]$. We plot distributions of $(\rho(A, B), \rho(C, D))$ in Panel (b) of Figure~\ref{fig:main}. 

We find that, among 10000 simulations, only 3 choice frequencies lie within the CRE area in Panel (c). In other words, there is almost zero chance of rejecting expected utility, even for a prospect theory preference that deviates substantially from expected utility, if we use the criterion of Panel (c). In contrast, under the strong paired choice test (see the CRE area in Panel (b) of Figure~\ref{fig:eu-cre-rcre}), expected utility is correctly rejected for 9727 of 10000 simulations. 

To further illustrate the low power of the test criterion in Panel (c), we note that it is very unlikely to generate a preference reversal consistent with the common ratio effect when the values of $x, y, p, r$ are arbitrarily chosen as in \mnoss.\footnote{As noted in \mnossc, the prior literature on the common ratio effect has focused on values such that $\frac{x}{p\,y}\in [0.75, 1]$.} For example, consider a CRRA utility function $u(x)=x^\gamma$ and distortion function $w(p)=\frac{p^\sigma}{(p^\sigma+(1-p)^\sigma)^{1/\sigma}}$ with $\gamma=\sigma=0.8$. Let us consider all parameter value combinations of $(x, y, p, r)$ such that $x, y\in \{1, 2, \ldots, 10\}$ with $x<y$, $p, r\in \{0.1, 0.2, \ldots, 0.9\}$. There are 3645 possible combinations of such values. However, only 383 combinations (10.5\%) exhibit a preference reversal consistent with the common ratio effect: $A=(x, 1)\succ B=(y, p)$ and $C=(x, r)\prec D=(y, r\,p)$. For the remaining 3262 combinations (89.5\%), comparisons are consistent with expected utility: either $A\succsim B$ and $C\succsim D$ or $A\prec B$ and $C\prec D$. 

In other words, we are unlikely to observe the common ratio effect, even for prospect theory preferences that the strong test would reject as being inconsistent with expected utility, for arbitrary combinations of the parameters in the paired choice task. The reason is that the expected values of $A$ and $B$, $x$ and $p\,y$ may be too far apart for a preference reversal. Conditional on exhibiting the common ratio effect (i.e., within the aforementioned 383 combinations), the average $p$ is $0.57$ and the average ratio $r$ is $0.44$ (the median $p$ is $0.6$ and the median $r$ is $0.4$). Moreover, for all the 383 combinations that produce the common ratio effect, we have $\frac{x}{p\,y}\in  [0.8, 1.2]$. To some extent, our exercise reinforces the message of \mnoss that finding the common ratio effect depends on where one looks for it.

The issue of parameter selection poses an important question. The implicit assumption in choosing arbitrary parameter values is that we want to understand how ``global'' the common ratio effect is; how prevalent it is across possible decision problems. Another perspective is that we want to know whether subjects at all exhibit the common ratio effect --- \emph{somewhere}. A formal definition of the effect in a deterministic setting is $(x, 1)\sim (y, p)$ and $(x, r)\prec (y, r\,p)$ for any $x, y, r, p$. The paired valuation test is based on this definition, and it is reasonable to expect to observe the effect for various combinations of  $x, y, r, p$ if the effect is robust. However, this is not the case for paired-choice tasks, where the same values of $x, y, r, p$ are used across many experimental subjects with different risk preferences. It is not possible to find $x, y, r, p$ such that $(x, 1)\sim (y, p)$ for many different subjects. Moreover, if the values are such that  $(x, 1)\prec (y, p)$, then we cannot find the effect, i.e., there is a bias against the common ratio effect. So experimentalists need to find values such that $(x, 1)\succ (y, p)$ and $(x, r)\prec (y, r\,p)$ if expected utility is the null hypothesis and the common ratio effect is the alternative hypothesis. Similarly, if $x$ is too large relative to $p\,y$, then we obtain $(x, 1)\succ (y, p)$ and $(x, r)\succ (y, r\,p)$; i.e., no choice reversal. It is arguably for this reason that the literature has focused on parameter values such that the expected values of $A$ and $B$, $x$ and $p\,y$, are close to each other, and the majority of subjects prefer $A$ over $B$; i.e., $\rho(A, B)\ge 1/2$.

\section{Implications}\label{sec:takeaway}

What is the takeaway message of our results? First, the assumed model of stochastic choice is very important. When we follow the decision theoretic literature and adopt the model of random expected utility, the weak paired choice test is unbiased and reinforces the standard conclusion in the literature that the common ratio effect is prevalent.

Second, if we instead focus on other models of stochastic choice (including the models of Fechnerian additive random utility that are common in the experimental literature), then the strong paired choice test is unbiased. Indeed, the literature has often focused on the strong paired choice test. \cite{kahnemann1979prospect} describe one of the earliest experimental demonstrations of the common ratio effect. Their result is that 80\% of subjects choose $A$ over $B$, while 65\% of them choose $D$ over $C$. They write, ``To show that the modal pattern of preferences in Problems 3 and 4 is not compatible with the theory,'' and then they proceed to explain why the patterns $\rho(A, B)>\frac{1}{2}$ and $\rho(C, D)<\frac{1}{2}$ are inconsistent with expected utility theory. \cite{ballinger1997}, who pioneered the literature testing the common ratio effect in the stochastic choice framework, defined the common ratio effect in this way. Under the unbiased strong paired choice test, we find significant support for the common ratio effect, as discussed in the previous section and in Figure~\ref{fig:main}.

Third, our results qualify the findings in \mnossc, who find no evidence of a systematic common ratio effect in their aggregate data. Their central conclusions contrast with the previous literature,  which is almost exclusively based on paired-choice tasks and shows strong evidence of a common ratio effect. \mnoss explain the discrepancy through three mechanisms. The first is that the weak paired choice test is biased towards finding a common ratio effect, while the paired valuation test is (under the same assumptions) unbiased. Second, using valuation tests, they do not find evidence of an aggregate common ratio effect. Third, their experiments with paired choice tests are validated at the individual level by the valuation test and indicate little evidence of a systematic common ratio effect. We offer a somewhat different perspective on these discrepancies. 

In Proposition~\ref{prop:anythingoes}, we demonstrate that for paired valuation tests under Assumption 2b, essentially ``anything goes.'' For a CRRA utility, the mean test is systematically biased unless participants are exactly risk neutral, and the sign test hinges on a demanding symmetry condition that tightly constrains the correlation of preference shocks across tasks. In sharp contrast, the strong paired choice test remains unbiased across a wide range of stochastic choice models and requires assumptions that are strictly weaker than Assumption 2b. We also show that the Panel (c) test for the weak paired choice condition has low power, making it difficult to reject the null of expected utility even when it is false.

Taken together, these findings provide robust empirical support for the Allais paradox in the guise of the common ratio effect. If we treat the random expected utility model as the canonical representation of expected utility in stochastic choice, the common ratio effect typically emerges through the weak paired choice test. If we instead want conclusions that are robust to alternative formulations of stochastic expected utility theory, the strong paired choice test is broadly unbiased. Under the more demanding strong test, we still find clear and substantial evidence for common ratio and reverse common ratio effects.

The strong paired choice test we advocate for can be applied more broadly and is useful for detecting various behavioral puzzles in choice data. Many behavioral puzzles take a form similar to the common ratio effect: $A$ is preferred over $B$, but $D$ is preferred over $C$, where $C$ and $D$ are certain transformations of $A$ and $B$, respectively. 
As long as $U(A)-U(B)\ge 0$ if and only if $U(C)-U(D)\ge 0$, where $U$ is a utility representation for a ``standard'' model we wish to test, our strong test ``$\rho(A, B)\ge 1/2$ iff $\rho(C, D)\ge 1/2$'' is robustly unbiased. For example, consider the \emph{common consequence effect} of the Allais paradox, which is exhibited if $A=(x, 1)$ is preferred over $B=(y, p; x, q)$ while $D=(y, p)$ is preferred over $C=(x, 1-q)$. Since $U(A)-U(B)=(1-q)\,u(x)-p\,u(y)=U(C)-U(D)$ under expected utility, the strong test is unbiased. Alternatively, consider \emph{present bias} in intertemporal choice. The bias is exhibited if $A=(x, 0)$ is preferred over $B=(y, t)$ while $D=(y, t+s)$ is preferred over $C=(x, s)$, where $(x, t)$ is a delayed reward that returns $x$ at time $t$. Since $U(A)-U(B)=u(x)-\delta^t\, u(y)=\frac{U(C)-U(D)}{\delta^s}$ under exponential discounting, the strong test is unbiased. Therefore, our strong test and the corresponding discussion of different approaches to modeling expected utility are readily applicable to various behavioral puzzles and choice contexts.

\section{Proofs}\label{sec:proofs}

\subsection{Proof of Proposition~\ref{prop:anythingoes}}

For notational economy,  define $B(\gamma) = p\, y^\gamma$. We need to study the following random valuations:
\begin{align}
   m_{\mathrm{AB}} =  \Gamma_1(\ep_{\mathrm{AB}}) &= \left( B(\gamma) + \ep_{\mathrm{AB}} \right)^{1/\gamma} \\
   m_{\mathrm{CD}} =  \Gamma_2(\ep_{\mathrm{CD}}) &= \left( B(\gamma) + \frac{\ep_{\mathrm{CD}}}{r} \right)^{1/\gamma}
\end{align}

\noindent\textbf{Part 1.} First, we prove the first statement in Proposition~\ref{prop:anythingoes}. Suppose that $\ep_{\mathrm{AB}}$ and $\ep_{\mathrm{CD}}$ are mean-zero, symmetric, random variables such that $\ep_{\mathrm{AB}} \stackrel{d}{=} k\ep_{\mathrm{CD}}$ for some $k > 0$. Specifically, we assume symmetric two-point distributions $\{-c_1, c_1\}$ and $\{-c_2, c_2\}$, subject to the constraint that the base of the exponent remains strictly positive (so $c_1 < B(\gamma)$ and $c_2 < r B(\gamma)$). 

We seek to determine the full range of achievable values for $(\E[\Gamma_1], \E[\Gamma_2])$ as $c_1, c_2$ span their permissible supports and as $\gamma$ varies from $1$ down to $0$.

Because $1/\gamma \ge 1$, the transformations of the two random variables are convex. By Jensen's Inequality, the minimum expected value for any mean-zero distribution occurs when the variance is zero. For both functions, this absolute minimum is:
$$ E_{min}(\gamma) = B(\gamma)^{1/\gamma} = \left(p \cdot y^\gamma\right)^{1/\gamma} = y\, \left(p^{1/\gamma}\right). $$

The maximum expected value is achieved by maximizing the variance. For $\Gamma_1$, as $c_1 \to B(\gamma)$:
$$ \lim_{c_1 \to B(\gamma)} \left[ \frac{1}{2}(B(\gamma) - c_1)^{1/\gamma} + \frac{1}{2}(B(\gamma) + c_1)^{1/\gamma} \right] = \frac{1}{2}(2B(\gamma))^{1/\gamma} $$
Substituting $B(\gamma)$ yields the supremum:
$$ E_{max}(\gamma) = 2^{\frac{1-\gamma}{\gamma}} \left(p \cdot y^\gamma\right)^{1/\gamma} = \frac{y}{2}(2p)^{1/\gamma} $$

Because the variances of $\ep_{\mathrm{AB}}$ and $\ep_{\mathrm{CD}}$ can be chosen independently via the scalar $k$, the joint achievable range for any fixed $\gamma$ is exactly the Cartesian product of their individual ranges:
$$ \mathcal{E}(\gamma)= (E_{min}(\gamma), \; E_{max}(\gamma)) \times (E_{min}(\gamma), \; E_{max}(\gamma)).$$

Note that the minimum is monotonically increasing in $\gamma$ and the maximum is monotonically decreasing as long as $2p>1$.

As $\gamma \to 1$, the limits converge to $E_{min}(1) = p\, y$ and $E_{max}(1) = \frac{y}{2}(2p) = p\, y$. The square $\mathcal{E}(\gamma)$ collapses to the point $(p\, y, p\, y)$. In contrast, as $\gamma \to 0$ we get  $\lim_{\gamma \to 0} y\left(p^{1/\gamma}\right) = 0$ and $\lim_{\gamma \to 0} \frac{y}{2}(2\, p)^{1/\gamma} = \infty$. So
\[ \bigcup_{\gamma\in (0,1)} \mathcal{E}(\gamma)= \Re^2_+.
\]

\smallskip
\noindent\textbf{Part 2.} We now turn to the proof of the second statement in Proposition~\ref{prop:anythingoes}. In particular, we now consider the values that
\[
\Pr(m_{\mathrm{AB}}>m_{\mathrm{CD}}) = \Pr(\Gamma_1(\ep_{\mathrm{AB}}) > \Gamma_2(\ep_{\mathrm{CD}}))
\] can take under the same assumptions as above, with the difference that we now fix $\ep_i$ to be uniform distributions because this allows for a richer behavior of the joint distribution of $(\ep_{\mathrm{AB}},\ep_{\mathrm{CD}})$.

As before, we may define $B(\gamma)=p\, y^{\gamma}$; the actual numbers of $y$, $p$ and $r$ do not matter, only that $B(\gamma)>0$ and $r>0$.  
Note that 
\begin{align*}
\Pr\left( (B(\gamma) + \ep_{\mathrm{AB}})^{1/\gamma} > \left(B(\gamma) + \frac{1}{r}\epsilon_{\mathrm{CD}}\right)^{1/\gamma} \right) &= \Pr\left( B(\gamma) + \ep_{\mathrm{AB}} > B(\gamma) + \frac{1}{r}\epsilon_{\mathrm{CD}} \right) \\
&= \Pr( \ep_{\mathrm{AB}} > \frac{1}{r}\epsilon_{\mathrm{CD}} ).
\end{align*}

We assume that $\ep_{\mathrm{AB}}$ and $\ep_{\mathrm{CD}}$ are two mean-zero uniform random variables, possibly with different supports. This ensures that $\ep_{\mathrm{AB}}$ equals $k \ep_{\mathrm{CD}}$ in distribution for some $k>0$. The supports need to be such that $B(\gamma)+\ep_{\mathrm{AB}}\geq 0$ and $B(\gamma)+\frac{1}{r}\ep_{\mathrm{CD}}\geq 0$ with probability one because the power $1/\gamma$ may not be defined otherwise.

Suppose first that the support of $\ep_{\mathrm{AB}}$ is greater than the support of $\frac{1}{r}\ep_{\mathrm{CD}}$. Define $X=\ep_{\mathrm{AB}}$ and $Y=\frac{1}{r}\ep_{\mathrm{CD}}$. Let $c$ and $d$ be positive real numbers such that $B(\gamma)\geq c > d$.  Let $Y \sim U[-d, d]$ and $Z \sim U[-c, c]$ be two independent, uniform, random variables. Define the random variable $X$ piecewise, conditional on the value of $Z$ and $Y$, by
\[      X = 
    \begin{cases} 
      Z & \text{if } Z < -d \\
      a + bY & \text{if } Z \geq -d 
   \end{cases}
\]
where $a = \frac{c-d}{2}$ and $b = \frac{c+d}{2d}$. This construction is illustrated in Figure~\ref{fig:pfProp7}.

\begin{figure}[t]
  \centering
  \begin{tikzpicture}[scale=.9, >=Stealth]
    % === 1. Draw the Axes ===
    \draw[->, thick] (-3.2,0) -- (3.2,0) node[right] {\footnotesize $Y$};
    \draw[->, thick] (0,-3.2) -- (0,3.2) node[above] {\footnotesize $X$};
    
    % Origin (Adjusted size for balance)
    \node[below right] at (0,0) {\scriptsize $(0,0)$};
    
    % === 2. Outer and Inner Boundaries ===
    % Draw the [-c, c]^2 square (dotted)
    \draw[dotted, thick] (-2.5,-2.5) rectangle (2.5,2.5);
    
    % Draw the [-d, d]^2 square (dotted)
    \draw[dotted, thick] (-1.25,-1.25) rectangle (1.25,1.25);
    
    % === 3. Ticks and Labels ===
    \foreach \coord/\label in {2.5/c, -2.5/-c, 1.25/d, -1.25/-d} {
        % X-axis ticks
        \draw (\coord, 0.05) -- (\coord, -0.05) node[below] {\footnotesize $\label$};
        % Y-axis ticks
        \draw (0.05, \coord) -- (-0.05, \coord) node[left] {\footnotesize $\label$};
    }
    
    % === 4. Blue Construction Line and Points ===
    % Line X = a + bY: (-d, -d) to (d, c)
    \draw[thick, blue] (-1.25,-1.25) -- (1.25,2.5);
    
    % Points at the ends of the segment
    \filldraw[blue] (-1.25,-1.25) circle (1.5pt);
    \filldraw[blue] (1.25,2.5) circle (1.5pt);
    
    % Line Label (moved slightly for clarity)
    \node[blue, above right] at (1.25, 2.5) {\footnotesize $X = a + bY$};
    
    % === 5. Gray Reference Diagonal ===
    % Diagonal X=Y (dashed, slightly shortened for aesthetics)
    \draw[dashed, gray!60] (-2.5,-2.5) -- (2.5,2.5);

    % === 6. Braces with Flipped Decoration ===

% Lower Brace: Draw bottom-to-top so it points LEFT
    \draw[decorate, decoration={brace, amplitude=5pt}, thin, red] 
        (-2.7, -2.5) -- (-2.7, -1.27) 
        node[midway, xshift=-15pt, rotate=90] {\footnotesize $X=Z$};

    % Upper Brace: Draw bottom-to-top so it points LEFT
    \draw[decorate, decoration={brace, amplitude=5pt}, thin, red] 
        (-2.7, -1.20) -- (-2.7, 2.5)
        node[midway, xshift=-15pt, rotate=90] {\footnotesize $\mbox{ }\qquad X=a+bY$};

  \end{tikzpicture}
  \caption{Illustration of the construction in the proof of Proposition~\ref{prop:anythingoes}.}
  \label{fig:pfProp7}
\end{figure}
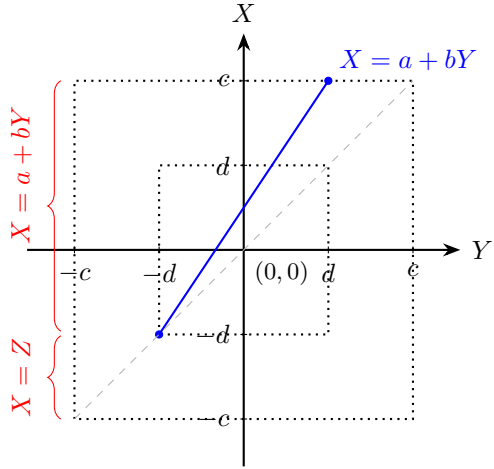

First, we show that $X\sim U[-c,c]$. To this end, we focus on the decomposition of $\Pr(X<x)$ according to the values of $Z$. Specifically, 
\[
    \Pr(X < x) = \Pr(X < x \wedge Z < -d) + \Pr(X < x \wedge Z \geq -d).
\]

\textbf{Case 1:} $x < -d$. Since the minimum value of $X$ conditional on $Z \geq -d$ is $-d$, the term $\Pr(X < x \wedge Z \geq -d)$  is zero for any $x < -d$. Conditional on $Z<-d$,  $X = Z$. Thus,
\begin{align*}
  \Pr(X < x) = \Pr(X < x \wedge Z < -d) & = \Pr(Z < x \wedge  Z < -d) \\
  & = \Pr(Z < x) = \frac{x - (-c)}{2c} = \frac{x+c}{2c},
  \end{align*} as $x<-d$ and $Z\sim U[-c,c]$.

\textbf{Case 2:} $x \in [-d, c]$. Again, we use the decomposition of $\Pr(X<x)$ according to $Z$. We have that  $x \geq -d$ and the event $Z < -d$ implies that $X=Z \leq x$. Therefore, the intersection of these events is simply $Z < -d$:
\[
\Pr(X < x \wedge Z < -d) = \Pr(Z \leq x \wedge Z < -d) = \Pr(Z < -d) = \frac{c-d}{2c}.
\]

For the second term, conditional on $Z \geq -d$, $X = a + bY$. By the independence of $Y$ and $Z$:
\[
\Pr(X < x \wedge Z \geq -d) = \Pr(a + bY < x) \Pr(Z \geq -d).
\]
Now, 
\[
\Pr(a + bY \leq x) = \Pr\left(Y \leq \frac{x-a}{b}\right) = \frac{\frac{x-a}{b} - (-d)}{2d} = \frac{x - a + bd}{2bd}
\]
Substituting $a = \frac{c-d}{2}$ and $b = \frac{c+d}{2d}$, we note that $bd = \frac{c+d}{2}$ and $a - bd = -d$. This simplifies the expression to $\frac{x+d}{c+d}$.  Multiplying this by $\Pr(Z \geq -d) = \frac{c+d}{2c}$ we obtain
\[\Pr(X \leq x \wedge Z \geq -d) = \left(\frac{x+d}{c+d}\right)\left(\frac{c+d}{2c}\right) = \frac{x+d}{2c}.
\]
Adding the two partitioned terms together gives the CDF for $x \in [-d, c]$:
\begin{equation*}
    \Pr(X \leq x) = \frac{c-d}{2c} + \frac{x+d}{2c} = \frac{x+c}{2c}
\end{equation*}

We conclude that $\Pr(X<x)$ equals the CDF of a uniform random variable on $[-c,c]$.

Finally, we calculate $\Pr(X>Y)$. Again we decompose this probability according to the possible values of $Z$:
\[
 \Pr(X > Y) = \Pr(X > Y \mid Z < -d)\Pr(Z < -d) + \Pr(X > Y \mid Z \geq -d)\Pr(Z \geq -d).
\]

First consider the case when $Z < -d$.  Here, $X = Z$. Since $Y$ has minimum value $-d$, and $Z < -d$, we have $X < -d \leq Y$. Thus $\Pr(X > Y \mid Z < -d) = 0$. 

Second, consider when $Z \geq -d$. We need to evaluate the inequality $X=a + bY > Y$, i.e., $\frac{a}{1-b}=d>Y$. Since $Y\sim U[-d,d]$, we have $d>Y$ with probability $1$. Geometrically, this is clear from Figure~\ref{fig:pfProp7}. Thus, $\Pr(X > Y \mid Z \geq -d) = 1$.

Putting these together, we obtain that 
\[
\Pr(X > Y) = (0)\left(\frac{c-d}{2c}\right) + (1)\left(\frac{c+d}{2c}\right) = \frac{c+d}{2c}\in (1/2,1].
\]
For each value of $q \in (1/2,1]$ we may find $d\in (0,c]$ such that $\Pr(X > Y) =q$.

Finally, we turn to the case when the support of $\ep_{\mathrm{AB}}$ is smaller than the support of $\frac{1}{r}\ep_{\mathrm{CD}}$. Then we can let  $Y=\ep_{\mathrm{AB}}$ and $X=\frac{1}{r}\ep_{\mathrm{CD}}$. Our previous analysis constructs a joint distribution with the desired marginals, and such that, for each $q \in (0,1/2)$ we may find $d\in (0,c]$ for which \[
\Pr(\ep_{\mathrm{AB}}>\frac{1}{r}\ep_{\mathrm{CD}}) = 1-\Pr(X> Y) = \frac{c-d}{2c}=q.
\]
The only value of $q$ not covered by the proof is $1/2$, but this is achieved trivially with independent errors.

\subsection{Proof of Proposition~\ref{prop:linearity}} Suppose there are functions $F$ and $u$ such that $\rho$ is scalable with respect to expected utility. It is without loss of generality to assume that $u(0)=0$. Suppose that Weak Linearity is satisfied. Note that Weak Linearity implies that
\[\rho((x, p), (y, p\,q))=\rho((x, r\,p), (y, r\,p\,q))\text{ for any }p\le 1.\]
For any $(x, r)$ and $(y, q)$, 
\[F(pu(x), p\,q\,u(y))=\rho((x, p), (y, p\,q))=\rho((x, r\,p), (y, r\,p\,q))=F(r\,p\,u(x), r\,p\,q\,u(y)).\]
Recall that $X$ is a finite set of monetary prizes such that $0\in X\subset \Re_{+}$. Let $u^*=u(x^*)=\max_{x\in X} u(x)$ and fix $\delta\in (0, u)$. 

Take any $b\le a\le u^*$. First, suppose that $a\ge \delta$. Since the range of $p\,u$ is connected, we can find $x, y, p, q, r$ such that $p\,u(x)=a$, $p\,q\,u(y)=b$, and $r\,p\,u(x)=\delta$. Hence, we have
\[F(a, b)=F(p\,u(x), p\,q\,u(y))=F(r\,p\,u(x), r\,p\,q\, u(y))=F(\delta, \delta\,\frac{b}{a}).\]

Second, suppose that $a<\delta$. Then we can find $x, y, p, q, r$ such that $p\, u(x)=\delta$, $p\,q\,u(y)=\delta\,\frac{b}{a}$, and $r\,p\,u(x)=a$. Hence, we have
\[F(a, b)=F(r\,p\,u(x), r\,p\,q\,u(y))=F(p\,u(x), p\,q\, u(y))=F(\delta, \delta\,\frac{b}{a}).\]
Hence, we have $F(a, b)=F(\delta, \delta\,\frac{b}{a})$ for any $b\le a\le u^*$. Since $F(a, b)=1-F(b, a)$, we also have $F(a, b)=F(\delta, \delta\,\frac{b}{a})$ for any $a, b\le u^*$.\footnote{Notice that $F(a, b)$ is redundant if either $a>u^*$ or $b>u^*$.}
Moreover, there is a function $G$ such that $F(\delta, \delta \frac{b}{a})=G(\ln(a)-\ln(b)).$ Hence, 

\[\rho(\ell_A, \ell_B)=F\big(\E u(\ell_A), \E u(\ell_B)\big)=G\big(\ln(\E u(\ell_A))- \ln(\E u(\ell_B))\big).\]
This proves the second part of the proposition. To prove the first part, let us now assume that linearity is satisfied. Then, by the proof of the second part, for any $\ell$, we need to have
\[\rho(\ell_A, \ell_B)=G\big(\ln(\E u(\ell_A))- \ln(\E u(\ell_B))\big)=\rho(\alpha\ell_A+(1-\alpha)\ell, \alpha \ell_B+(1-\alpha)\ell)\]
\[=G\big(\ln(\alpha \E u(\ell_A)+(1-\alpha)\E u(\ell))- \ln(\alpha\E u(\ell_B)+(1-\alpha)\E u(\ell))\big),\]
which is violated because $\frac{a+c}{b+c}\neq \frac{a}{b}$ for any $a, b, c>0$ with $a\neq b$.

\subsection{Proof of Proposition~\ref{prop:weakEU}}

Since $F(0)=\frac{1}{2}$, we have 
\[\rho(A, B)=F\big(G(A, B)(u(x)-p\,u(y))\big)\ge \frac{1}{2}\] iff $u(x)\ge p\,u(y)$ iff $\rho(C, D)=F\big(G(C, D)(r\,u(x)-r\,p\,u(y))\big)\ge \frac{1}{2}$. 

\subsection{Proof of Proposition~\ref{prop:assumption2}} Recall that $\Gamma$ is strictly increasing and $m^*_{\mathrm{AB}}=m^*_{\mathrm{CD}}$. Let $x$ be the monetary quantity in Lottery $A$. Hence, if $x\ge \Gamma(m^*_{\mathrm{AB}}, \epsilon')$, then we have 
\[\rho(A, B)=\Pr(x\ge \Gamma(m^*_{\mathrm{AB}}, \ep_{\mathrm{AB}}))\ge \Pr(\Gamma(m^*_{\mathrm{AB}}, \epsilon')\ge \Gamma(m^*_{\mathrm{AB}}, \ep_{\mathrm{AB}}))=\Pr(\epsilon'\ge \ep_{\mathrm{AB}})=\frac{1}{2}\] 
and
\[\rho(C, D)=\Pr(x\ge \Gamma(m^*_{\mathrm{CD}}, \ep_{\mathrm{CD}}))\ge \Pr(\Gamma(m^*_{\mathrm{CD}}, \epsilon')\ge \Gamma(m^*_{\mathrm{CD}}, \ep_{\mathrm{CD}}))=\Pr(\epsilon'\ge \ep_{\mathrm{CD}})=\frac{1}{2}.\] 

Similarly, if $x\le \Gamma(m^*_{\mathrm{AB}}, \epsilon')$,  then we have 
\[\rho(A, B)=\Pr(x\ge \Gamma(m^*_{\mathrm{AB}}, \ep_{\mathrm{AB}}))\le \Pr(\Gamma(m^*_{\mathrm{AB}}, \epsilon')\ge \Gamma(m^*_{\mathrm{AB}}, \ep_{\mathrm{AB}}))=\Pr(\epsilon'\ge \ep_{\mathrm{AB}})=\frac{1}{2}\] 
and
\[\rho(C, D)=\Pr(x\ge \Gamma(m^*_{\mathrm{CD}}, \ep_{\mathrm{CD}}))\le \Pr(\Gamma(m^*_{\mathrm{CD}}, \epsilon')\ge \Gamma(m^*_{\mathrm{CD}}, \ep_{\mathrm{CD}}))=\Pr(\epsilon'\ge\ep_{\mathrm{CD}})=\frac{1}{2}.\] 
Hence, $\rho(A, B)\ge \frac{1}{2}$ if and only if $\rho(C, D)\ge \frac{1}{2}$.

\subsection{Proof of Proposition~\ref{prop:stpairchoiceab}}

By Equation (1), 

\[\rho(A, B)=\Pr(u(x)-p\,u(y)>\xi(A, B)),\]
where
\[\xi(A, B)=\epsilon_p\,f(y)+\epsilon_y\,g(p)-\epsilon_1\,f(x)-\epsilon_x\,g(1)+(\epsilon_{p, y}-\epsilon_{1, x}),\]
and 
\[\rho(C, D)=\Pr(r(u(x)-p\,u(y))>\xi(C, D)),\]
where
\[\xi(C, D)=\epsilon_{pr}\,f(y)+\epsilon_y\,g(pr)-\epsilon_r\,f(x)-\epsilon_x\,g(r)+(\epsilon_{pr, y}-\epsilon_{r, x}).\]
It is enough to show that $\xi(A, B)$ and $\xi(C, D)$ are symmetric around zero. Let us prove that $\xi(C, D)$ is symmetric around zero for all three cases (essentially identical arguments will imply that $\xi(A, B)$ is also symmetric around zero). 
We repeatedly use the following two facts:

\smallskip
\noindent\textbf{Fact 1:} $X$ and $Y$ are symmetric around zero and either independent or linearly dependent. Then $a X+bY$ is symmetric around zero for any $a, b\in\Re$.

If $X$ and $Y$ are linearly dependent (and since they are symmetric around zero), $X=c\,Y$ for some $c$. Hence, $a\,X+b\,Y=(a+bc)\,Y$ is symmetric around zero. If $X$ and $Y$ are independent, then $a X+bY$ is also symmetric around zero. 

\smallskip
\noindent\textbf{Fact 2:} $X$ and $Y$ are symmetric around zero and independent. Then $a X+bY+c\,X\,Y$ is symmetric around zero for any $a, b, c\in\Re$.

Notice that we wrote $\xi(C, D)$ as a weighted sum of five error terms. All three cases directly assume that each of the first four terms is symmetric around zero. Note that the fifth term $\epsilon_{pr, y}-\epsilon_{r, x}$ is also symmetric around zero. For the first case, $\epsilon_{pr, y}$ and $\epsilon_{r, x}$ are symmetric around zero and either independent or linearly dependent. Hence, $\epsilon_{pr, y}-\epsilon_{r, x}$ is symmetric around zero. For the second  case, $\epsilon_{pr, y}$ and $\epsilon_{r, x}$ are symmetric around $\gamma$ because $\epsilon_p$ and $\epsilon_x$ are symmetric around zero and independent. For the third case, $\epsilon_{pr, y}-\epsilon_{r, x}$ is symmetric around zero because $\epsilon_{pr, y}$ and $\epsilon_{r, x}$ are independent, and either identical or symmetric around the same fixed constant. 

Finally, $\xi(C, D)$ is symmetric around zero by Fact 1 in the first and third cases, while by Fact 2 in the second case.

\subsection{Proof of Proposition~\ref{prop:pvtestbiased}} Let us assume that errors $\epsilon_p$ are i.i.d. and symmetric around zero and $\epsilon_{p, x}=\epsilon_x=0$ with probability $1$. This case satisfies all three assumptions of Proposition~\ref{prop:stpairchoiceab}. Then by the proof of part one of Proposition 10, we have $\E[\ep_{\mathrm{AB}}]\neq \E[\ep_{\mathrm{CD}}]$. It is not difficult to provide examples with non-degenerate error distributions.

\clearpage

\bibliographystyle{ecta}
\bibliography{allais}

\begin{thebibliography}{29}
\newcommand{\enquote}[1]{``#1''}
\expandafter\ifx\csname natexlab\endcsname\relax\def\natexlab#1{#1}\fi

\bibitem[\protect\citeauthoryear{Agranov and Ortoleva}{Agranov and
  Ortoleva}{2017}]{agranov2017stochastic}
\textsc{Agranov, M. and P.~Ortoleva} (2017): \enquote{Stochastic choice and
  preferences for randomization,} \emph{Journal of Political Economy}, 125,
  40--68.

\bibitem[\protect\citeauthoryear{Allais}{Allais}{1953}]{allais1953comportement}
\textsc{Allais, M.} (1953): \enquote{Le comportement de l'homme rationnel
  devant le risque: critique des postulats et axiomes de l'{\'e}cole
  am{\'e}ricaine,} \emph{Econometrica: journal of the Econometric Society},
  503--546.

\bibitem[\protect\citeauthoryear{Apesteguia and Ballester}{Apesteguia and
  Ballester}{2018}]{apesteguia2018monotone}
\textsc{Apesteguia, J. and M.~A. Ballester} (2018): \enquote{Monotone
  stochastic choice models: The case of risk and time preferences,}
  \emph{Journal of Political Economy}, 126, 74--106.

\bibitem[\protect\citeauthoryear{Ballinger and Wilcox}{Ballinger and
  Wilcox}{1997}]{ballinger1997}
\textsc{Ballinger, T.~P. and N.~T. Wilcox} (1997): \enquote{Decisions, Error
  and Heterogeneity,} \emph{The Economic Journal}, 107, 1090--1105.

\bibitem[\protect\citeauthoryear{Barberis}{Barberis}{2013}]{barberis2013thirty}
\textsc{Barberis, N.~C.} (2013): \enquote{Thirty years of prospect theory in
  economics: A review and assessment,} \emph{Journal of economic perspectives},
  27, 173--196.

\bibitem[\protect\citeauthoryear{Ben-Akiva}{Ben-Akiva}{1973}]{ben1973structure}
\textsc{Ben-Akiva, M.~E.} (1973): \enquote{Structure of passenger travel demand
  models.} Ph.D. thesis, Massachusetts Institute of Technology.

\bibitem[\protect\citeauthoryear{Berry, Levinsohn, and Pakes}{Berry
  et~al.}{1995}]{berry1995automobile}
\textsc{Berry, S., J.~Levinsohn, and A.~Pakes} (1995): \enquote{Automobile
  Prices in Market Equilibrium,} \emph{Econometrica}, 63, 841--890.

\bibitem[\protect\citeauthoryear{Blavatskyy, Panchenko, and Ortmann}{Blavatskyy
  et~al.}{2023}]{Blavatskyy_Panchenko_Ortmann_2023}
\textsc{Blavatskyy, P., V.~Panchenko, and A.~Ortmann} (2023): \enquote{How
  common is the common-ratio effect?} \emph{Experimental Economics}, 26,
  253–272.

\bibitem[\protect\citeauthoryear{Camerer}{Camerer}{1995}]{camerer1995individual}
\textsc{Camerer, C.} (1995): \enquote{Individual decision making,} in \emph{The
  handbook of experimental economics}, ed. by J.~H. Kagel and A.~E. Roth,
  Princeton University Press, 587--704.

\bibitem[\protect\citeauthoryear{Fishburn}{Fishburn}{1978}]{fishburn1978choice}
\textsc{Fishburn, P.~C.} (1978): \enquote{Choice probabilities and choice
  functions,} \emph{Journal of Mathematical Psychology}, 18, 205--219.

\bibitem[\protect\citeauthoryear{Gul and Pesendorfer}{Gul and
  Pesendorfer}{2006}]{gul2006random}
\textsc{Gul, F. and W.~Pesendorfer} (2006): \enquote{Random expected utility,}
  \emph{Econometrica}, 74, 121--146.

\bibitem[\protect\citeauthoryear{Hausman and Wise}{Hausman and
  Wise}{1978}]{hausman1978conditional}
\textsc{Hausman, J.~A. and D.~A. Wise} (1978): \enquote{A conditional probit
  model for qualitative choice: Discrete decisions recognizing interdependence
  and heterogeneous preferences,} \emph{Econometrica: Journal of the
  econometric society}, 403--426.

\bibitem[\protect\citeauthoryear{He and Natenzon}{He and
  Natenzon}{2024}]{he2024moderate}
\textsc{He, J. and P.~Natenzon} (2024): \enquote{Moderate utility,}
  \emph{American Economic Review: Insights}, 6, 176--195.

\bibitem[\protect\citeauthoryear{Hey}{Hey}{2001}]{hey2001does}
\textsc{Hey, J.~D.} (2001): \enquote{Does repetition improve consistency?}
  \emph{Experimental economics}, 4, 5--54.

\bibitem[\protect\citeauthoryear{Kahneman and Tversky}{Kahneman and
  Tversky}{1979}]{kahnemann1979prospect}
\textsc{Kahneman, D. and A.~Tversky} (1979): \enquote{Prospect Theory. An
  Analysis of Decision under Uncertainty,} \emph{Econometrica}, 47, 263--291.

\bibitem[\protect\citeauthoryear{Loomes}{Loomes}{2005}]{Loomes_2005}
\textsc{Loomes, G.} (2005): \enquote{Modelling the Stochastic Component of
  Behaviour in Experiments: Some Issues for the Interpretation of Data,}
  \emph{Experimental Economics}, 8, 301–323.

\bibitem[\protect\citeauthoryear{Machina}{Machina}{1987}]{machina1987choice}
\textsc{Machina, M.~J.} (1987): \enquote{Choice under uncertainty: Problems
  solved and unsolved,} \emph{Journal of Economic Perspectives}, 1, 121--154.

\bibitem[\protect\citeauthoryear{Machina}{Machina}{2008}]{machina2008non}
---\hspace{-.1pt}---\hspace{-.1pt}--- (2008): \enquote{Non-expected utility
  theory,} in \emph{The New Palgrave Dictionary of Economics}, Springer, 1--14.

\bibitem[\protect\citeauthoryear{Machina}{Machina}{2018}]{machina2018non}
---\hspace{-.1pt}---\hspace{-.1pt}--- (2018): \enquote{Non-expected utility
  theory,} in \emph{The New Palgrave Dictionary of Economics}, Springer,
  9570--9582.

\bibitem[\protect\citeauthoryear{McFadden}{McFadden}{1978}]{mcfadden1978modeling}
\textsc{McFadden, D.} (1978): \enquote{Modeling the Choice of Residential
  Location,} \emph{Spatial Interaction Theory and Planning Models}, 75--96.

\bibitem[\protect\citeauthoryear{McGranaghan, Nielsen, O’Donoghue,
  Somerville, and Sprenger}{McGranaghan
  et~al.}{2024}]{mcgranaghan2024distinguishing}
\textsc{McGranaghan, C., K.~Nielsen, T.~O’Donoghue, J.~Somerville, and C.~D.
  Sprenger} (2024): \enquote{Distinguishing common ratio preferences from
  common ratio effects using paired valuation tasks,} \emph{American Economic
  Review}, 114, 307--347.

\bibitem[\protect\citeauthoryear{Nevo}{Nevo}{2000}]{nevo2000}
\textsc{Nevo, A.} (2000): \enquote{A Practitioner's Guide to Estimation of
  Random-Coefficients Logit Models of Demand,} \emph{Journal of Economics \&
  Management Strategy}, 9, 513--548.

\bibitem[\protect\citeauthoryear{Prelec}{Prelec}{1998}]{prelec1998}
\textsc{Prelec, D.} (1998): \enquote{The Probability Weighting Function,}
  \emph{Econometrica}, 66, 497--527.

\bibitem[\protect\citeauthoryear{Strzalecki}{Strzalecki}{2025}]{strzalecki2025stochastic}
\textsc{Strzalecki, T.} (2025): \enquote{Stochastic choice theory,}
  \emph{Cambridge Books}.

\bibitem[\protect\citeauthoryear{Thurstone}{Thurstone}{1927}]{thurstone1927psychophysical}
\textsc{Thurstone, L.~L.} (1927): \enquote{Psychophysical analysis,} \emph{The
  American journal of psychology}, 38, 368--389.

\bibitem[\protect\citeauthoryear{Tversky}{Tversky}{1969}]{tversky1969intransitivity}
\textsc{Tversky, A.} (1969): \enquote{Intransitivity of preferences.}
  \emph{Psychological review}, 76, 31.

\bibitem[\protect\citeauthoryear{Tversky}{Tversky}{1972}]{tversky1972elimination}
---\hspace{-.1pt}---\hspace{-.1pt}--- (1972): \enquote{Elimination by aspects:
  A theory of choice.} \emph{Psychological review}, 79, 281.

\bibitem[\protect\citeauthoryear{Wilcox}{Wilcox}{2008}]{wilcox2008stochastic}
\textsc{Wilcox, N.~T.} (2008): \enquote{Stochastic models for binary discrete
  choice under risk: A critical primer and econometric comparison,} in
  \emph{Risk aversion in experiments}, Emerald Group Publishing.

\bibitem[\protect\citeauthoryear{Wilcox}{Wilcox}{2011}]{WILCOX201189}
---\hspace{-.1pt}---\hspace{-.1pt}--- (2011): \enquote{‘Stochastically more
  risk averse:’ A contextual theory of stochastic discrete choice under
  risk,} \emph{Journal of Econometrics}, 162, 89--104, the Economics and
  Econometrics of Risk.

\end{thebibliography}

\clearpage

\appendix

\section{Assumptions 2b and 3}\label{sec:assmn2b3}

Here, we provide a more elaborate example where Assumption 2b holds while Assumption 3 is violated. \mnoss define the symmetry assumption for bivariate random variables with a density. So here we provide an example that has a (strictly positive) density.

Let $(\ep_{\mathrm{AB}}, \ep_{\mathrm{CD}})$ be a bivariate random vector defined on the support $[-1, 1]^2$ with the following joint probability density function:
\[
f(z_1, z_2) = \frac{1}{4} + \frac{1}{4} z_1 \left(z_2^2 - \frac{1}{3}\right).\]

\begin{proposition}\label{prop:assmn2band3}
    The marginal distributions of $\ep_{\mathrm{AB}}$ and  $\ep_{\mathrm{CD}}$ are uniform on  $[-1, 1]$, but Assumption 3 is violated: the property of central symmetry, $f(z_1, z_2) = f(-z_1, -z_2)$ does not hold. 
\end{proposition}

\begin{proof}[\textbf{Proof of Proposition~\ref{prop:assmn2band3}}] Observe first that  $f(z_1, z_2)$ is a strictly positive density function on $[-1, 1]^2$. Indeed, for 
  $z_1 \in [-1, 1]$ and $z_2^2 - \frac{1}{3} \in \left[-\frac{1}{3}, \frac{2}{3}\right]$ the minimum value of the term $z_1 \left(z_2^2 - \frac{1}{3}\right)$ occurs when $z_1 = -1$ and $z_2 = \pm 1$, yielding a value of $-\frac{2}{3}$. Substituting this into the density function provides its global minimum over $[-1,1]^2$: 
  \[
  \frac{1}{4} - \frac{1}{6} = \frac{1}{12} > 0.
  \]
  
Second, we show that the marginal distributions are standard uniform. To find the marginal density of $\ep_{\mathrm{AB}}$, we integrate the joint density over the support of $z_2$:
\begin{align*}f_{\ep_{\mathrm{AB}}}(z_1) &= \int_{-1}^{1} \left[ \frac{1}{4} + \frac{1}{4} z_1 \left(z_2^2 - \frac{1}{3}\right) \right] dz_2 \\
&= \frac{1}{2} + \frac{1}{4} z_1 \left( \left(\frac{1}{3} - \frac{1}{3}\right) - \left(-\frac{1}{3} + \frac{1}{3}\right) \right) 
= \frac{1}{2}.
\end{align*}
Thus, $\ep_{\mathrm{AB}} \sim U[-1, 1]$. That $\ep{_\mathrm{CD}} \sim U[-1, 1]$ is immediate by integrating the density.

Finally, note that 
\begin{align*}f(1, 1) = \frac{1}{4} + \frac{1}{4}(1)\left(1^2 - \frac{1}{3}\right) = \frac{5}{12} \neq 
\frac{1}{4} + \frac{1}{4}(-1)\left((-1)^2 - \frac{1}{3}\right) = f(-1, -1) 
\end{align*}
Thus, we have $f(z_1, z_2) \neq f(-z_1, -z_2)$ and the symmetry property in Assumption 3 is violated. 
\end{proof}

\clearpage

\section{Biasedness of the mean valuation test}\label{sec:meantestbias}

\mnoss show that, under their Assumption 2a, the mean valuation test is unbiased. Under the iAREU model, however, Assumption 2a imposes very strong restrictions on the utility function. The assumption does not hold, even in the iAREU case with errors that are additive and i.i.d. Note that the equality $\Gamma(m, \epsilon)=m+\epsilon$ is equivalent to $u(m)+\epsilon=u(m+\epsilon)$ as $\Gamma(m, \epsilon)=u^{-1}\big(u(m)+\epsilon\big)$. However, the equality $u(m)+\epsilon=u(m+\epsilon)$ cannot be satisfied unless $u$ is linear. In fact, even the weaker condition $\E[u(m)+\epsilon]=\E[u(m+\epsilon)]$ cannot generally be satisfied. For example, when $u$ is strictly concave and $\E[\epsilon]=0$, we have $\E[u(m)+\epsilon]=\E[u(m)]>\E[u(m+\epsilon)]$. 

Consistent with this observation, the mean valuation test is always biased when $u$ is strictly concave or strictly convex.  

In order for the mean test to be unbiased, it must be that 
\[\E[\Delta m]=\E\big[h(p\,u(y)+\ep_{\mathrm{CD}})-h(p\,u(y)+\ep_{\mathrm{AB}})\big]=0,\]
where $h=u^{-1}$, under the hypothesis of additive random utility.

By the mean value theorem, 
\[h(p\,u(y)+\ep_{\mathrm{CD}})-h(p\,u(y)+\ep_{\mathrm{AB}})=(\ep_{\mathrm{CD}}-\ep_{\mathrm{AB}})\cdot h'(p\,u(y)+\delta(\ep_{\mathrm{CD}}, \ep_{\mathrm{AB}})),\]
where $\delta(\ep_{\mathrm{CD}}, \ep_{\mathrm{AB}})\in [\min(\ep_{\mathrm{CD}}, \ep_{\mathrm{AB}}), \max(\ep_{\mathrm{CD}}, \ep_{\mathrm{AB}})]$. Unless $h$ is linear (meaning that $u$ is linear), $\delta(\ep_{\mathrm{CD}}, \ep_{\mathrm{AB}})$ is a non-constant random variable that depends on both $\ep_{\mathrm{AB}}$ and $\ep_{\mathrm{CD}}$. 

Hence, the mean test is theoretically valid (or unbiased) only if 
\[\E[(\ep_{\mathrm{CD}}-\ep_{\mathrm{AB}})\cdot h'(p\,u(y)+\delta(\ep_{\mathrm{CD}}, \ep_{\mathrm{AB}}))]=0.\]

Observe that $\E[XY]=\Cov(X, Y)+\E[X]\,\E[Y]$ for any random variables $X$ and $Y$. So for the paired valuation test to be unbiased, we require that 
\[\Cov\Big(\ep_{\mathrm{CD}}-\ep_{\mathrm{AB}}, h'\big(p\,u(y)+\delta(\ep_{\mathrm{CD}}, \ep_{\mathrm{AB}})\big)\Big)=0,\] as $\E[\ep_{\mathrm{CD}}-\ep_{\mathrm{AB}}]=0$ by our assumption that errors are mean zero.

Our first observation is that  $h'\big(p\,u(y)+\delta(\ep_{\mathrm{CD}}, \ep_{\mathrm{AB}})\big)$ and $\ep_{\mathrm{CD}}-\ep_{\mathrm{AB}}$ can be correlated, and therefore the paired valuation test biased, even when individual errors $\epsilon_{x,p}$ are i.i.d. and symmetric around zero. To this end, we present two simple examples, which we then generalize in Proposition~\ref{prop:pvaltest}.

\begin{example}
    Let $u(x)=\sqrt{x}$ and the errors $\epsilon(\ell)$ be i.i.d. Note that $\ep_{\mathrm{CD}}=\frac{\epsilon_D-\epsilon_C}{r}$ and $\ep_{\mathrm{AB}}=\epsilon_B-\epsilon_A$ are independent and symmetric random variables, and there is $k>0$ such that $\ep_{\mathrm{CD}}\overset{d}{=}k\,\ep_{\mathrm{AB}}$ (i.e., Assumption 2b is satisfied). Note that $h(t)=t^2$. Hence, 
\[h'\big(p\,u(y)+\delta(\ep_{\mathrm{CD}}, \ep_{\mathrm{AB}})\big)=2\,p\,u(y)+\ep_{\mathrm{CD}}+\ep_{\mathrm{AB}}\text{ and }\delta(\ep_{\mathrm{CD}}, \ep_{\mathrm{AB}})=\frac{\ep_{\mathrm{CD}}+\ep_{\mathrm{AB}}}{2}.\]
Consequently, 
\[\E[\Delta m]=\E[(\ep_{\mathrm{CD}}-\ep_{\mathrm{AB}})\,(2\,p\,u(y)+\ep_{\mathrm{CD}}+\ep_{\mathrm{AB}})]=\frac{1-r^2}{r^2}\text{Var}[\ep_{\mathrm{AB}}]>0.\]
\end{example}

\begin{example}
 Let $u(x)=x^2$, and let $\epsilon(\ell)$ are independently and uniformly distributed on $\{-1, 1\}$. Note that $\ep_{\mathrm{CD}}=\frac{\epsilon_D-\epsilon_C}{r}$ and $\ep_{\mathrm{AB}}=\epsilon_B-\epsilon_A$ are independent and symmetric random variables, and there is $k>0$ such that $\ep_{\mathrm{CD}}\overset{d}{=}k\,\ep_{\mathrm{AB}}$ (i.e., Assumption 2b is satisfied). Note that $h(t)=\sqrt{t}$. Let $p=r=0.2$, and $y=10$. Then,  
\[\E[\Delta m]=\E[\sqrt{p\,u(y)+\ep_{\mathrm{CD}}}-\sqrt{p\,u(y)+\ep_{\mathrm{AB}}}]\approx -0.07<0.\]    
\end{example}

Let us generalize the above two examples.

\begin{prps}\label{prop:pvaltest} Suppose that $\epsilon(\ell)$ are i.i.d. 
\begin{enumerate}
    \item If $u$ is strictly concave, then $\E[\Delta m]>0$. 
    \item If $u$ is strictly convex, then $\E[\Delta m]<0$.
\end{enumerate}
\end{prps}

A similarly negative result is included in the online appendix to \mnoss and is mentioned in the main text of the paper. 

\begin{proof}[\textbf{Proof of Proposition~\ref{prop:pvaltest}}] Let $X=p\,u(y)+\ep_{\mathrm{CD}}$ and $Y=p\,u(y)+\ep_{\mathrm{AB}}$. Note that $\ep_{\mathrm{CD}}\overset{d}{=}\frac{1}{r}\,\ep_{\mathrm{AB}}$ and $\ep_{\mathrm{AB}}=\epsilon_B-\epsilon_A$ is symmetric about zero. Therefore, $X$ is a mean-preserving spread of $Y$ (which will be shown below). Hence, if $h$ is strictly concave (i.e., $u$ is strictly convex), we have 
\[\E[\Delta m]=\E[h(X)]-\E[h(Y)]<0.\]
Similarly, if $h$ is strictly convex (i.e., $u$ is strictly concave), we have 
\[\E[\Delta m]=\E[h(X)]-\E[h(Y)]>0.\]
Finally, let us show that $X$ is a mean-preserving spread of $Y$, equivalently, $\ep_{\mathrm{CD}}$ is a mean-preserving spread of $\ep_{\mathrm{AB}}$, i.e., there is $Z$ such that $\ep_{\mathrm{CD}}\overset{d}=\ep_{\mathrm{AB}}+Z$ and $\E[Z \mid \ep_{\mathrm{AB}}=x]=0$ for all $x.$ For simplicity, we prove for the discrete case. Let us construct $Z$ as follows: for each $x$, $Z \mid \ep_{\mathrm{AB}}=x$ is equal to $(\frac{1}{r}-1)x$ with probability $\frac{1+r}{2}$ and $(\frac{1}{r}+1)(-x)$ with probability $\frac{1-r}{2}$. Note that 
\[\E[Z \mid \ep_{\mathrm{AB}}=x]=(\frac{1}{r}-1)x\,\frac{1+r}{2}+(\frac{1}{r}+1)(-x)\frac{1-r}{2}=\frac{x(1-r^2-(1-r^2))}{2r}=0.\]

Suppose that $\ep_{\mathrm{AB}}+Z=x/r$ and $\ep_{\mathrm{AB}}=y$. There are two possibilities: either $\ep_{\mathrm{AB}}-y(\frac{1}{r}+1)=x/r$ or $\ep_{\mathrm{AB}}+(\frac{1}{r}-1)y=x/r$.
In the first case, $y-y(\frac{1}{r}+1)=x/r$, so $y=-x$. In the second case, 
$y+(\frac{1}{r}-1)y=x/r$, so $y=x$. Hence, 
\[Pr(\ep_{\mathrm{AB}}+Z=\frac{1}{r}\,x)=\Pr(\ep_{\mathrm{AB}}=x, Z=(\frac{1}{r}-1)x)+\Pr(\ep_{\mathrm{AB}}=-x, Z=(\frac{1}{r}+1)x)\]
\[=\Pr(\ep_{\mathrm{AB}}=x)(\frac{1+r}{2})+\Pr(\ep_{\mathrm{AB}}=-x)(\frac{1-r}{2})=\Pr(\ep_{\mathrm{AB}}=x),\]
the last equality holds because $\ep_{\mathrm{AB}}$ is symmetric about zero. Hence, $\ep_{\mathrm{CD}}\overset{d}=\ep_{\mathrm{AB}}+Z$.
\end{proof}

\section{Biasedness of the sign valuation test}\label{sec:pairedval}

In this appendix, we provide further discussion on the sign valuation test.  We show that the sign test can be biased under natural assumptions regarding the model of stochastic choice. We present results that complement the findings in Proposition~\ref{prop:anythingoes}, and we complete the table that was presented in the introduction.

First, we complement the ``anything goes'' message of Proposition~\ref{prop:anythingoes}. The proof of the second statement in the proposition assumes correlated errors. Here we show that, in the model of Equation~\ref{eq:perception}, we still obtain a result in the same spirit as Proposition~\ref{prop:anythingoes}.

For simplicity, we assume that $\epsilon_1=0$ with probability 1. We construct error distributions for $\epsilon_s$, $s\in \{p,r,rp\}$, and for $\epsilon_z$, $z\in \Re_+$, for which the sign test can be biased either in the direction of the common ratio effect or the reverse common ratio effect. 

\begin{proposition}\label{prop:signedtestprosp}
  Let $\theta\in (0,1)$ and $p+r\neq 1$.
  There are independent mean-zero distributions $F$, $F'$, $G$, and $G'$, so that if $\epsilon_s \sim F$ and $\epsilon_z\sim G$, then
 \[\Pr(m_{\mathrm{CD}}>m_{\mathrm{AB}})>\theta; \] and if 
  $\epsilon_s \sim F'$ and $\epsilon_z\sim G'$, then
  \[\Pr(m_{\mathrm{AB}}>m_{\mathrm{CD}})>\theta. \]
\end{proposition}

\begin{remark}
    The proof of Proposition~\ref{prop:signedtestprosp} assumes errors $\epsilon_z$ that are degenerate, but it is easy to extend the construction to non-degenerate errors on $u$.
\end{remark}

\begin{proof}[\textbf{Proof of Proposition~\ref{prop:signedtestprosp}}] Suppose that $p+r<1$. Let $\gamma\in (0,1)$ be small enough that $(1-\gamma)^3>\theta$. Suppose that $\epsilon_z=0$ with probability one and let $\epsilon_s$ equal $a$ with probability $1-\gamma$ and $-a\frac{1-\gamma}{\gamma}$ with probability $\gamma$, for $s=p,r,rp$. Then $\epsilon_s$ has mean zero. Now choose $a>0$ small enough that
  \[-a\frac{1-\gamma}{\gamma}<r \quad\text{ and }\quad a<1-(p+r). \] Then $r+\epsilon_s>0$ with probability one, and we have the following sequence of implications:
  \begin{align*}
pr + a(1-(p+r)) & > pr+a^2 \\
\then pr + a & > (p+a)(r+a) \\
\then u(y)\frac{pr+a}{r+a} & >u(y) (p+a) \\
\then m_{\mathrm{CD}}=h[u(y)\frac{pr+a}{r+a}] & >m_{\mathrm{AB}}=h[u(y) (p+a)], \\
  \end{align*} where $h=u^{-1}$; we have used that $u(y)>0$ and that $h$ is strictly increasing. This means that, when $\epsilon_s=a$ for all $s=r,pr,p$, then we have that $m_{\mathrm{CD}}>m_{\mathrm{AB}}$. But this occurs with probability $(1-\gamma)^3>\theta$.

To obtain the opposite inequality, we may choose $\gamma$ and $a\in (0,1-(p+r))$ as before and small enough that $r-a>0$. Define $\epsilon_s=-a$ with probability $1-\gamma$ and $\epsilon_s=a\frac{1-\gamma}{\gamma}$ with complementary probability. Then it is easy to see that $m_{\mathrm{AB}}>m_{\mathrm{CD}}$ when $\epsilon_s=-a$, which occurs with probability $(1-\gamma)^3>\theta$.

An analogous construction works for the case when $p+r>1$.

\end{proof}

\subsection{Sign valuation test under Equation (1)}\label{sec:signtest} 

To prove the unbiasedness of the paired valuation sign-test, \mnoss impose assumptions that all imply $\E[\ep_{\mathrm{AB}}] = \E[\ep_{\mathrm{CD}}]$. Here, we explain why this assumption is violated under the random utility model given by Equation (1). Consequently, it may be difficult to justify the model with Assumption 2 (or the assumptions behind Proposition 2 in their paper) beyond iAREU.

For expositional simplicity, let us consider the following special case of Equation (1): the random utility of a simple lottery $(x,p)$ is $\tilde v(p)\tilde u(x)$, where $\tilde v(p) = p+\epsilon_p$ and $\tilde u(x)= u(x)+\epsilon_x$ are such that $\epsilon_p$ and $\epsilon_x$ have mean zero for all $p$ and $x$. We do assume that $\tilde v(p)>0$ for all $p$; otherwise, we might incur a violation of monotonicity. We consider both cases where (i) $\epsilon_1=0$, and (ii) for $p<1$, $\epsilon_1$ and $\epsilon_p$ have identical distributions. We shall also single out the case when $\tilde v(1)=v(1)=1$ occurs with probability 1 (i.e., $\epsilon_1=0$), reflecting a common assumption about probability weighting functions. 

Before turning to the details, note that the utility of a lottery $(x,p)$ is now
\begin{equation}
\tilde v(p)\tilde u(x)=(p+\epsilon_p)[u(x)+\epsilon_x]= p\,u(x) + \underbrace{ \epsilon_p\,u(x)+\epsilon_x\,p+\epsilon_p\,\epsilon_x}_{\text{``additive'' error}},
\end{equation} 
which means that there is some important structure hidden behind the assumption of an additive error. Any assumption on such errors masks an underlying assumption on the model of random utility that can be hard to evaluate.

The valuation tasks give us $m_{\mathrm{AB}}$ and $m_{\mathrm{CD}}$ such that:
\begin{align*}
   (1+\epsilon_1)\,(u(m_{\mathrm{AB}})+\epsilon_{m_{\mathrm{AB}}}) & =(p+\epsilon_p)\,(u(y)+\epsilon^1_{y}) \\
   (r+\epsilon_r)\,(u(m_{\mathrm{CD}})+\epsilon_{m_{\mathrm{CD}}}) & =(pr+\epsilon_{pr})\,(u(y)+\epsilon^2_{y}).
\end{align*}
Hence, 
\[m_{\mathrm{AB}}=u^{-1}\big(p\,u(y)+\ep_{\mathrm{AB}}\big)\text{ and }m_{\mathrm{CD}}=u^{-1}\big(p\,u(y)+\ep_{\mathrm{CD}}\big),\]
where
\[\ep_{\mathrm{AB}}=p\,\epsilon^1_y-\epsilon_{m_{\mathrm{AB}}}+u(y)\,\zeta_{\mathrm{AB}}+\epsilon^1_{y}\,\zeta_{\mathrm{AB}}\text{ and }\zeta_{\mathrm{AB}}=\frac{\epsilon_{p}-p\,\epsilon_1}{1+\epsilon_1}.\]
and
\[\ep_{\mathrm{CD}}=p\,\epsilon^2_y-\epsilon_{m_{\mathrm{CD}}}+u(y)\,\zeta_{\mathrm{CD}}+\epsilon^2_y\,\zeta_{\mathrm{CD}}\text{ and }\zeta_{\mathrm{CD}}=\frac{\epsilon_{pr}-p\,\epsilon_r}{r+\epsilon_r}.\]
Again, $\ep_{\mathrm{AB}}$ and $\ep_{\mathrm{CD}}$ are the residual terms that are fundamentally different from errors $\epsilon_p$ and $\epsilon_x$. To provide a simple intuition for why $\E[\ep_{\mathrm{AB}}]\neq \E[\ep_{\mathrm{CD}}]$, notice that 
\[\E[\ep_{\mathrm{AB}}]=\E[p\,\epsilon^1_y-\epsilon_{m_{\mathrm{AB}}}+u(y)\,\zeta_{\mathrm{AB}}+\epsilon^1_{y}\,\zeta_{\mathrm{AB}}]=u(y)\,\E[\zeta_{\mathrm{AB}}]+\text{Cov}(\epsilon^1_{y}, \zeta_{\mathrm{AB}}).\]
and 
\[\E[\ep_{\mathrm{CD}}]=\E[p\,\epsilon^2_y-\epsilon_{m_{\mathrm{CD}}}+u(y)\,\zeta_{\mathrm{CD}}+\epsilon^2_y\,\zeta_{\mathrm{CD}}]=u(y)\,\E[\zeta_{\mathrm{CD}}]+\text{Cov}(\epsilon^2_y, \zeta_{\mathrm{CD}}).\]

It turns out that $\E[\zeta_{\mathrm{CD}}]\neq 0$ unless $\epsilon_r=0$, i.e., we are in the case of random expected utility of Section~\ref{sec:REU}. Hence, the preference error $\epsilon_y$ and ``perceptual,'' or probability weighting, errors $\epsilon_r$ have to be correlated in a particular fashion so that
\[\Cov(\epsilon^2_y, \zeta_{\mathrm{CD}})=-u(y)\,\E[\zeta_{\mathrm{CD}}]\]
is satisfied. Hence, unless such knife-edge equality is satisfied, we cannot have $\E[\ep_{\mathrm{AB}}]=\E[\ep_{\mathrm{CD}}]$. Let us consider the three specific cases in the following result. 

\begin{proposition}\label{prop:epABepCD} Assume Equation (1).
\begin{enumerate}
\item  If errors are independent and either $\epsilon_1=0$ for sure, or $\epsilon_1 \overset{d}{=} \epsilon_r$, then  $\E[\ep_{\mathrm{AB}}]<\E[\ep_{\mathrm{CD}}]$.
\item 
If $\epsilon_p=\epsilon_r$ (``perfect correlation''),  $\epsilon_p$ and $\epsilon_x$ are independent, and  $\epsilon_1=0$ or  $\epsilon_1 \overset{d}{=} \epsilon_r$, then $\E[\ep_{\mathrm{AB}}]>\E[\ep_{\mathrm{CD}}]$.
\item 
If $\epsilon_p=\epsilon_r=\lambda\,\epsilon_x$ (``perfect correlation''), then:
\begin{enumerate}
\item 
When $\epsilon_1=0$, and $u(y)>\frac{r}{\lambda}$, then $\E[\ep_{\mathrm{AB}}]>0>\E[\ep_{\mathrm{CD}}]$.
\item 
When $\epsilon_1=\epsilon_r$ and  
$\frac{1}{\lambda}>u(y)>\frac{r}{\lambda}$, then $\E[\ep_{\mathrm{AB}}]>0>\E[\ep_{\mathrm{CD}}]$.
\end{enumerate}
\end{enumerate}
In all three cases, $\E[\ep_{\mathrm{AB}}]\neq \E[\ep_{\mathrm{CD}}]$. Moreover, if all errors are symmetric around zero, then
\[\rho(A, B)\ge \frac{1}{2}\text{ and }\rho(C, D)\ge \frac{1}{2}.\]
\end{proposition}

\begin{remark}
In all the cases covered by Proposition~\ref{prop:epABepCD}, imposing  $\E[\ep_{\mathrm{AB}}]=\E[\ep_{\mathrm{CD}}]=0$ reduces the model to the random expected utility theory we analyzed in Section~\ref{sec:REU}.
\end{remark}

\begin{proof}[\textbf{Proof of Proposition~\ref{prop:epABepCD}}] To prove the first part, let us assume that all errors are independent. Then $\text{Cov}(\epsilon^1_y, \zeta_{\mathrm{AB}})=\text{Cov}(\epsilon^2_y, \zeta_{\mathrm{CD}})=0$. Hence, we have 
\[\E[\ep_{\mathrm{CD}}]=u(y)\,\E[\zeta_{\mathrm{CD}}]=u(y)\E[\frac{\epsilon_{pr}-p\,\epsilon_r}{r+\epsilon_r}]\]
\[=u(y)\, (\E[\frac{\epsilon_{pr}+pr}{r+\epsilon_r}]-p)=p\,u(y)\, (\E[\frac{r}{r+\epsilon_r}]-1),\]
where we have used that $\E[\frac{\epsilon_{pr}}{r+\epsilon_r}]=0$ by independence and mean zero. Then, since $z\mapsto \frac{r}{r+z}$ is a strictly convex function as long as $r+z>0$, we have $\E[\frac{r}{r+\epsilon_r}]>1$ by Jensen's inequality. Thus, $\E[\ep_{\mathrm{CD}}]>0$. 

When $\epsilon_1=0$, we have $\E[\ep_{\mathrm{AB}}]=0$. When $\epsilon_1$ and $\epsilon_r$ have identical distributions, we have
\[\E[\ep_{\mathrm{AB}}]=p\,u(y)\, (\E[\frac{1}{1+\epsilon_1}]-1).\]
Note that $\E[\frac{r}{r+\epsilon_r}]>\E[\frac{1}{1+\epsilon_1}]>1$. Hence, in either case, we have $\E[\ep_{\mathrm{AB}}]<\E[\ep_{\mathrm{CD}}]$.

Note that $\E[\ep_{\mathrm{AB}}]=\E[\ep_{\mathrm{CD}}]=0$ implies that $\epsilon_r=0$, i.e., the stochastic choice follows the random expected utility of Gul and Pesendorfer (2006).

To prove the second part, let us now assume that $\epsilon_p=\epsilon_r$ and $\epsilon_p$ and $\epsilon_x$ are independent. In this case, we still have  $\text{Cov}(\epsilon^1_y, \zeta_{\mathrm{AB}})=\text{Cov}(\epsilon^2_y, \zeta_{\mathrm{CD}})=0$. Hence, 

\[\E[\ep_{\mathrm{CD}}]=u(y)\,\E[\zeta_{\mathrm{CD}}]=u(y)\E[\frac{(1-p)\epsilon_r}{r+\epsilon_r}]\]
\[=(1-p)\,u(y)\, \E[\frac{\epsilon_r}{r+\epsilon_r}]=(1-p)\,u(y)\, (1-\E[\frac{r}{r+\epsilon_r}])<0.\]
Since $\frac{r}{r+x}$ is a strictly convex function as long as $r+x>0$, we have $\E[\frac{r}{r+\epsilon_r}]>1$ by Jensen's inequality. Hence, $\E[\ep_{\mathrm{CD}}]<0$.

When $\epsilon_1=0$, we have $\E[\ep_{\mathrm{AB}}]=0$. When $\epsilon_1$ and $\epsilon_r$ have identical distributions, we have $\E[\frac{r}{r+\epsilon_r}]>\E[\frac{1}{1+\epsilon_1}]>1$. Hence, in either case we have $\E[\ep_{\mathrm{AB}}]>\E[\ep_{\mathrm{CD}}]$. Again, in this case, $\E[\ep_{\mathrm{AB}}]=\E[\ep_{\mathrm{CD}}]=0$ implies that $\epsilon_r=0$, i.e., the stochastic choice follows the random expected utility.

To prove the third part, let us now assume that $\epsilon_p=\epsilon_r=\lambda\,\epsilon_x$. In this case, we have  
\[\text{Cov}(\epsilon^2_y, \zeta_{\mathrm{CD}})=\frac{(1-p)}{\lambda}\,\E[\frac{\epsilon^2_r}{r+\epsilon_r}]=\frac{(1-p)}{\lambda}\,\E[\frac{\epsilon^2_r-r^2+r^2}{r+\epsilon_r}]=\]
\[=\frac{(1-p)}{\lambda}\,(\E[\epsilon_r]-r+\E[\frac{r^2}{r+\epsilon_r}]=\frac{(1-p)\,r}{\lambda}\,(\E[\frac{r}{r+\epsilon_r}]-1)\]
and \[\E[\zeta_{\mathrm{CD}}]=(1-p)\E[\frac{\epsilon_r}{r+\epsilon_r}]=
(1-p)(1-\E[\frac{r}{r+\epsilon_r}]).\]
Hence, 
\[\E[\ep_{\mathrm{CD}}]=u(y)\,\E[\zeta_{\mathrm{CD}}]+\text{Cov}(\epsilon^2_y, \zeta_{\mathrm{CD}})=(1-p)(\E[\frac{r}{r+\epsilon_r}]-1)(\frac{r}{\lambda}-u(y)).\]
Since $(1-p)(\E[\frac{r}{r+\epsilon_r}]-1)>0$, $\E[\ep_{\mathrm{CD}}]\le 0$ iff $\frac{r}{\lambda}\le u(y)$.

When $\epsilon_1=0$, we have $\ep_{\mathrm{AB}}=\epsilon_p$. Then 
\[\E[\ep_{\mathrm{AB}}]=Cov(\epsilon^1_y, \epsilon_p)=\frac{Var(\epsilon_p)}{\lambda}>0.\]
Hence, if $u(y)>\frac{r}{\lambda}$, then $\E[\ep_{\mathrm{AB}}]>0>\E[\ep_{\mathrm{CD}}]. $

When $\epsilon_1=\epsilon_r$ , we have  
\[\text{Cov}(\epsilon^1_y, \zeta_{\mathrm{AB}})=\frac{(1-p)}{\lambda}\,\E[\frac{\epsilon^2_p}{1+\epsilon_p}]=\frac{(1-p)}{\lambda}\,\E[\frac{\epsilon^2_p-1+1}{1+\epsilon_p}]=\]
\[=\frac{(1-p)}{\lambda}\,(\E[\epsilon_p]-1+\E[\frac{1}{1+\epsilon_p}]=\frac{(1-p)}{\lambda}\,(\E[\frac{1}{1+\epsilon_p}]-1)\]
and \[\E[\zeta_{\mathrm{AB}}]=(1-p)\E[\frac{\epsilon_p}{1+\epsilon_p}]=
(1-p)(1-\E[\frac{1}{1+\epsilon_p}]).\]
Hence, 
\[\E[\ep_{\mathrm{AB}}]=u(y)\,\E[\zeta_{\mathrm{AB}}]+\text{Cov}(\epsilon^1_y, \zeta_{\mathrm{AB}})=(1-p)(\E[\frac{1}{1+\epsilon_p}]-1)(\frac{1}{\lambda}-u(y)).\]
Since $(1-p)(\E[\frac{1}{1+\epsilon_r}]-1)>0$, $\E[\ep_{\mathrm{AB}}]>0$ iff $\frac{1}{\lambda}>u(y)$. Hence, if $\frac{1}{\lambda}>u(y)>\frac{r}{\lambda}$, then $\E[\ep_{\mathrm{AB}}]>0>\E[\ep_{\mathrm{CD}}]. $

Again, in this case, $\E[\ep_{\mathrm{AB}}]=\E[\ep_{\mathrm{CD}}]=0$ implies that $\epsilon_p=0$, i.e., the stochastic choice follows the random expected utility.

Finally, suppose now that all errors are symmetric around zero. Note that the second assumption of Proposition~\ref{prop:stpairchoiceab} is satisfied. Hence, we obtain the desired inequalities by Proposition~\ref{prop:stpairchoiceab}. 
\end{proof}

\subsection{Prelec's probability weighting function}\label{sec:prelec}

Stochastic choice is often modeled  by means of a random coefficients specification. See, for example, the survey by \cite{nevo2000} of the methodology introduced by \cite{berry1995automobile}. The question is whether valuation tests are unbiased when we take a random deviation from expected utility theory that is modeled through a random coefficients specification. In particular, we start from a parametric version of prospect theory \citep{kahnemann1979prospect}, and take the parameters to signify random deviations from expected utility theory.

In prospect theory, the utility of a lottery that pays $x$ with probability $p$ and $0$ with probability $1-p$ is $v(p)u(x)$ when we normalize $u(0)=0$. Suppose, moreover, that utility takes the constant relative risk aversion form with a coefficient of relative risk aversion $\gamma$. The valuation tasks now give us:
\begin{align*}
    m_{\mathrm{AB}}^\gamma & = v(p) x^{\gamma} \\
    v(r)m_{\mathrm{CD}}^\gamma & = v(pr) x^{\gamma} \\
\end{align*}

First, we could let $v(p)=p$ and model random choice by allowing $\gamma$ to be random. The resulting model will be a special case of random expected utility, which we have already discussed.

Second, suppose instead that we take the function $v$ to be random. More specifically, suppose that it is random but ``centered'' on expected utility to capture random, unbiased deviations from expected utility. Let $\tilde v$ denote the random $v$. One draw of random $v$ determines $m_{\mathrm{AB}}$ and another draw determines $m_{\mathrm{CD}}$. 

Hence, $m_{\mathrm{AB}} < m_{\mathrm{CD}}$ if and only if 
\[
\tilde v(p)^{1/\gamma}<
\left(\frac{\tilde v(pr)}{\tilde v(r)}\right)^{1/\gamma}
\]
 
We should emphasize here that each elicitation involves a different realization of the probability weighting function. So, $\tilde v(p)$ is one draw, and $\tilde v(pr)$ and $\tilde v(r)$ are obtained from a second draw of $\tilde v$ evaluated at two different points. This is important because we are going to draw the parameters of the function $v$ at random. The first draw of $\tilde v$ corresponds to one realization of the random coefficients. The second draw corresponds to a second realization.

Now suppose that $v$ takes the form proposed by \cite{prelec1998}. So $v(p)=e^{-(-\ln (p))^\alpha}$ with $\alpha>0$. When $\alpha=1$, we have expected utility. When $\alpha<1$ ($\alpha>1$), we obtain the (reverse) common ratio effect. 

\begin{proposition}\label{prop:prelec}
Let $\tilde \alpha$ be drawn from a full-support probability distribution on $[1-\Delta,1+\Delta]$ that is symmetric around $1$. Suppose that $p\neq 1/e$.
\begin{enumerate}
    \item If $p<1/e$, then $\Pr(m_{\mathrm{CD}}<m_{\mathrm{AB}})<\frac{1}{2}$.
    \item If $p>1/e$, then $\Pr(m_{\mathrm{CD}}>m_{\mathrm{AB}})>\frac{1}{2}$.
\end{enumerate}
In contrast, if $\tilde \alpha$ is drawn from a probability distribution on $[1-\Delta,1+\Delta]$ with median $=1$, then the strong paired choice test is unbiased. 
\end{proposition}

\begin{remark}
Proposition~\ref{prop:prelec} means that, even though the deviations from expected utility are ``unbiased'' in the sense of being equally likely on each side of expected utility, the valuation test gives a biased conclusion regarding the common ratio effect. We should note that the proof allows for correlated draws of the random coefficients, and that the unbiasedness of the strong paired choice test only requires that the median of $\tilde \alpha$ is 1, not that its distribution be symmetric.

Finally, we should note that the second statement in the proposition is the most relevant case for the existing experimental literature. In 97\% of the studies reported by \cite{Blavatskyy_Panchenko_Ortmann_2023}, it holds that $p>1/e$. Moreover, in 10 out of 15 experimental cases of \mnossc, they have $p>\frac{1}{e}$, in which the sign test is biased against the common ratio effect. 
\end{remark}

\begin{proof}[\textbf{Proof of Proposition~\ref{prop:prelec}}]

Throughout the proof, we use the notation $a=\ln(1/p)$ and $b=\ln(1/r)$, and assume that $a\neq 1$ (i.e., $p\neq \frac{1}{e}$). Note that $a,b>0$. Under the assumptions we have made, $m_{\mathrm{AB}} < m_{\mathrm{CD}}$ if and only if \[0>
\ln m_{\mathrm{AB}} - \ln m_{\mathrm{CD}} = \frac{1}{\gamma}  [\ln \tilde v(p) - \ln \tilde v(pr) + \ln \tilde v(r)].
\] 

Note that
\begin{align*}
  \ln \tilde v(p) - \ln \tilde v(pr) + \ln \tilde v(r)
& = - (\ln(1/p))^{\tilde \alpha_1} + (\ln (1/p) + \ln(1/r))^{\tilde \alpha_2} - (\ln(1/r))^{\tilde \alpha_2}  \\
& = (a+b)^x - b^x - a^y,
\end{align*}
where $x=\tilde \alpha_1$ is the draw relevant for $m_{\mathrm{AB}}$ and $y=\tilde \alpha_2$ the one relevant for $m_{\mathrm{CD}}$. Let
\[
y(x) = \frac{1}{\ln(a)}\ln\left[(a+b)^x - b^x  \right]
\] be the value of $y$ for which we get $ \ln \tilde v(p) - \ln \tilde v(pr) - \ln \tilde v(r)=0$. 

Let $c = a+b$ and $d = b$ and consider $f(x) = \ln(c^x - d^x)$. Note that $c>d$. Note also that the sign of $f''$ equals the sign of $y''$ if and only if $\ln(a)>0$, which occurs exactly when $p<1/e$. For the remainder of the proof, we assume that $\ln(a)>0$ and obtain the first statement in Proposition~\ref{prop:prelec}. The second statement follows by reversing the signs. 

So we have that
\[
f''(x) = \frac{\left(c^x (\ln c)^2 - d^x (\ln d)^2\right) \left(c^x - d^x\right) - \left(c^x \ln c - d^x \ln d\right)^2}{\left(c^x - d^x\right)^2}
\]
The denominator is positive. We only need to study the numerator, $N(x)$.

\begin{align*} 
N(x) 
&= \left[ c^{2x} (\ln c)^2 - c^x d^x (\ln c)^2 - c^x d^x (\ln d)^2 + d^{2x} (\ln d)^2 \right] \\ 
&\quad - \left[ c^{2x} (\ln c)^2 - 2 c^x d^x \ln c \ln d + d^{2x} (\ln d)^2 \right]  \\
&= - c^x d^x (\ln c)^2 - c^x d^x (\ln d)^2 + 2 c^x d^x \ln c \ln d \\ 
&= - c^x d^x \left[ (\ln c)^2 + (\ln d)^2 - 2 \ln c \ln d \right] 
\end{align*}

Thus,
\[N(x) = -c^x d^x (\ln c - \ln d)^2<0\] as $c>d$.
Conclude that $f(x)$, and therefore $y(x)$, is strictly concave. 

Now consider a pair $(\tilde x,\tilde y)\in [1-\Delta,1+\Delta]\times [1-\Delta,1+\Delta]$. If $\tilde y>y(\tilde x)$ then $(a+b)^{\tilde x} - b^{\tilde x}< a^{\tilde y}$ and therefore $\ln \tilde v(p) - \ln \tilde v(pr) - \ln \tilde v(r)<0$. If $\tilde y<y(\tilde x)$ then $(a+b)^{\tilde x} - b^{\tilde x}> a^{\tilde y}$ and therefore $\ln \tilde v(p) - \ln \tilde v(pr) - \ln \tilde v(r)>0$. Since $y(1)=1$ and $y(\cdot)$ is strictly concave, its graph lies below the tangent line at $y(1)$. This line bisects the square $[1-\Delta,1+\Delta]\times [1-\Delta,1+\Delta]$. So the hypograph of $y(\cdot)$ is a proper subset of the area below this tangent. Thus, the area of $ [1-\Delta,1+\Delta]\times [1-\Delta,1+\Delta]$ that lies above the graph of $y(\cdot)$ is strictly greater than the area that lies below. 

Hence, for any symmetric distribution of $(\tilde x,\tilde y)$ on $[1-\Delta,1+\Delta]\times [1-\Delta,1+\Delta]$, we have that the probability that $\ln m_{\mathrm{AB}} - \ln m_{\mathrm{CD}} = \ln \tilde v(p) - \ln \tilde v(pr) - \ln \tilde v(r)<0$ is strictly greater than the probability that $\ln \tilde v(p) - \ln \tilde v(pr) - \ln \tilde v(r)>0$.

To prove the final statement regarding the strong paired choice test, let $T=\frac{\gamma \ln(y/x)}{\ln(1/p)}$. Note that 
\begin{align*}
\rho(A,B) & =\Pr[u(x)>\tilde v(p)u(y)]  = \Pr[(\frac{x}{y})^\gamma>\exp(-(\ln(1/p))^{\tilde \alpha})]\\
&=\Pr[\ln(1/p)^{\tilde{\alpha}-1}>\frac{\gamma \ln(y/x)}{\ln(1/p)}]=\Pr[a^{\tilde{\alpha}-1}>T].
\end{align*}
Similarly, 
\begin{align*}
\rho(C,D) & =\Pr[\gamma\ln(x)-(\ln(1/r))^{\tilde \alpha}>\gamma\ln(y)-(\ln(1/pr))^{\tilde \alpha}] \\
& =\Pr[\frac{(\ln(1/pr))^{\tilde \alpha}-(\ln(1/r))^{\tilde \alpha}}{\ln(1/p)}>\frac{\gamma \ln(y/x)}{\ln(1/p)}]\\
& =\Pr[\frac{(a+b)^{\tilde \alpha} - b^{\tilde \alpha}}{a} >T ].
\end{align*} 
We shall show that $\rho(A, B)\ge \frac{1}{2}$ if and only if $\rho(C, D)\ge \frac{1}{2}$. Note that $a^{\alpha-1}=\frac{(a+b)^{\alpha} - b^{\alpha}}{a}=1$ when $\alpha=1$.  Hence, it is enough to show that 
\[\Pr[a^{\tilde{\alpha}-1}>1]=\frac{1}{2}=\Pr[\frac{(a+b)^{\tilde \alpha} - b^{\tilde \alpha}}{a} >1].\]
To prove the first equality, note that $a^{\alpha-1}$ is strictly monotonic in $\alpha$ (either increasing or decreasing). Hence, when $a>1$, $a^{\alpha-1}>1$ for all $\alpha\in (1, 1+\Delta)$ and $a^{\alpha-1}<1$ for all $\alpha\in [1-\Delta, 1)$. When $a<1$, $a^{\alpha-1}<1$ for all $\alpha\in (1, 1+\Delta)$ and $a^{\alpha-1}>1$ for all $\alpha\in [1-\Delta, 1)$. We obtain $\Pr[a^{\tilde{\alpha}-1}>1]=\frac{1}{2}$ in either case.

To prove the second equality, let us first prove the following useful fact.

\medskip
\noindent\textbf{Fact:} If $g(x)=(a+b)^x-b^x-a$ is strictly monotonic on $[1-\Delta, 1+\Delta]$, then  $\Pr[\frac{(a+b)^{\tilde \alpha} - b^{\tilde \alpha}}{a} >1]=\frac{1}{2}$.

\smallskip
Since $g(1)=0$ and $g$ is strictly monotonic, we have either $g(x)>0$ for all $x\in (1, 1+\Delta]$ and $g(x)<0$ for all $x\in [1-\Delta, 1)$ or $g(x)<0$ for all $x\in (1, 1+\Delta]$ and $g(x)>0$ for all $x\in [1-\Delta, 1)$. Hence, we obtain $\Pr[\frac{(a+b)^{\tilde \alpha} - b^{\tilde \alpha}}{a} >1]=\frac{1}{2}$. \smallskip

We now consider two cases. 

\smallskip
\noindent\textbf{Case 1.} $a+b\ge 1$. 

\smallskip
In this case, $g(x)$ is strictly increasing. To see this, note that $g'(x)=(a+b)^x\,\ln(a+b)-b^x\,\ln(b)$. Since $a+b\ge 1$, we have $(a+b)^x\,\ln(a+b)\ge 0>b^x\,\ln(b)$. Hence, we obtain the desired result from the fact. 

\medskip
\noindent\textbf{Case 2.} $a+b<1$. 

\smallskip
In this case, 
\[g'(x)=(a+b)^x \ln(a+b)-b^x\ln(b)\ge 0 \text{ iff }x\ge x^*:=\frac{\ln\big(\frac{\ln(1/b)}{\ln(1/(a+b))}\big)}{\ln(\frac{a+b}{b})}>0.\]
If $x^*\not\in [1-\Delta, 1+\Delta]$, then $g$ is strictly monotonic on $[1-\Delta, 1+\Delta]$. Hence, we obtain the desired result from the fact. 

We now assume $x^*\in [1-\Delta, 1+\Delta]$. Note that it is not possible to have $x^*>1$. If $x^*>1$, then $g$ is strictly decreasing on $[0, 1]\subset [0, x^*]$. However, we have $g(0)=-a<g(1)=0$, a contradiction.

Hence, we should have $x^*\in [1-\Delta, 1]$. In this case, since $g$ is strictly increasing on $[x^*, 1+\Delta]$ and $g(1)=0$, we have $g(x)>0$ for any $x\in (1, 1+\Delta]$ and $g(x)<0$ for any $x\in [x^*, 1)$. Moreover, since $g$ is strictly decreasing on $(0, x^*)$ and $g(0)=-a<0$, we also have $g(x)<0$ for any $x\in [1-\Delta, x^*)$. To sum up, we have $g(x)>0$ for any $x\in (1, 1+\Delta]$ and $g(x)<0$ for any $x\in [1-\Delta, 1)$, which gives the desired result. 

\end{proof}

\end{document}